\documentclass[11pt]{article}

\usepackage{fullpage}

\usepackage[utf8]{inputenc}

\usepackage[algo2e,ruled,vlined,linesnumbered]{algorithm2e}

\usepackage[normalem]{ulem}
\usepackage[inline]{enumitem}
\usepackage{amsmath,amsfonts,amssymb,amsthm,hyperref,algorithm,eqparbox,array,ragged2e,stmaryrd,bbm,bm}
\usepackage{xspace}
\usepackage{tikz}
\usepackage{xcolor}
\usepackage{ifthen}
\usepackage{mathtools} 
\usepackage{bussproofs}
\usepackage{thmtools,thm-restate}

\usepackage{lineno}
\usepackage{thm-restate}
\usepackage{enumitem}
\usepackage[capitalise,nameinlink,noabbrev]{cleveref}

\usepackage[dvipsnames]{xcolor}
\hypersetup{
	colorlinks=true,
	urlcolor=MidnightBlue,
	linkcolor=MidnightBlue,
	citecolor=MidnightBlue,
	unicode
}

\declaretheorem[name=Theorem, parent=section]{theorem}
\declaretheorem[name=Lemma,sibling=theorem]{lemma}
\declaretheorem[name=Claim,sibling=theorem]{claim}
\declaretheorem[name=Proposition,sibling=theorem]{proposition}

\theoremstyle{definition}
\declaretheorem[name=Definition, sibling=theorem, style=definition]{definition}
\declaretheorem[name=Open Problem, sibling=theorem, style=definition]{problem}

\theoremstyle{remark}

\makeatletter
\def\th@example{%
  \thm@notefont{}%
  \normalfont %
}
\def\th@definition{%
  \thm@notefont{}%
  \normalfont %
}
\makeatother

\theoremstyle{example}

\input{restate}

\renewcommand{\setminus}{\smallsetminus}
\renewcommand{\hat}{\widehat}

\newcommand{\Z}{\mathbb{Z}}
\newcommand{\R}{\mathbb{R}}

\newcommand{\Th}{\mathrm{Th}}

\newcommand{\calR}{\mathcal{R}}
\newcommand{\calG}{\mathcal{G}}

\newcommand{\calT}{\mathcal{T}}

\newcommand{\NP}{\mathbf{NP}}

\DeclareMathOperator{\Exp}{\mathbb{E}}

\renewcommand{\epsilon}{\varepsilon}

\newcommand{\fix}{\textsf{\upshape fix}}

\newcommand{\search}{\mathrm{Search}}

\newcommand{\threshold}{q}

\newcommand{\random}[1]{\bm{#1}}
\newcommand{\newmeasure}[2]{\newcommand{#1}{{\textup{\sffamily #2}}\xspace}}
\newmeasure{\C}{C}

\newcommand{\cl}{\mathrm{Cl}}
\newcommand{\CNbound}{R}

\newcommand{\Hmin}{\mathrm{H}_{\infty}}

\newcommand{\xbf}{\bm{x}}
\newcommand{\ybf}{\bm{y}}
\newcommand{\vbf}{\bm{v}}

\newcommand{\Gbf}{\bm{G}}
\newcommand{\protocolPi}{\Pi}
\newcommand{\protocolPiAlt}{\Pi'}
\newcommand{\protocolPiRand}{\Pi_r}
\newcommand{\Hminof}[1]{\Hmin( {#1} )}
\newcommand{\commCost}[1]{\lvert#1\rvert}

\newcommand{\ResParity}{\mathrm{Res}(\oplus)}

\newcommand{\comment}[1]{}

\ifthenelse{1>0}{
\newcommand{\artur}[1]{\textcolor
{orange}{[Artur: #1]}}
\newcommand{\david}[1]{\textcolor{red}{[David: #1]}}
\newcommand{\duri}[1]{\textcolor{purple}{[Duri: #1]}}
\newcommand{\yassine}[1]{\textcolor{ForestGreen}{[Yassine: #1]}}
\newcommand{\susanna}[1]{\textcolor{blue}{[Susanna: #1]}}
}{
\newcommand{\artur}[1]{}
\newcommand{\david}[1]{}
\newcommand{\duri}[1]{}
\newcommand{\yassine}[1]{}
\newcommand{\susanna}[1]{}
}

\newcommand{\setsize}[1]{\lvert#1\rvert}

\newcommand{\prflength}{L}
\newcommand{\ResLin}{\textrm{Res}(\oplus)}
\newcommand{\binClique}[2]{\mathsf{BClique}({#1}, {#2})}
\newcommand{\binCliqueGk}{\binClique{G}{k}}
\newcommand{\binCliqueBlock}[2]{\mathsf{BClique}_{\textrm{block}}({#1}, {#2})}
\newcommand{\binCliqueBlockGk}{\binCliqueBlock{G}{k}}

\newcommand{\searchComm}[1]{\mathsf{BClique}^{cc}({#1}, k)}

\newcommand{\binCliqueBlockComm}[2]
{\mathsf{BClique}_{\textrm{block}}^{cc}(#1,#2)}

\newcommand{\graphSize}{N}

\newcommand{\kClique}{$k$-\textsf{Clique}\xspace}

\providecommand{\erdosrenyi}{Erd\H{o}s-R\'{e}nyi\xspace}

\makeatletter
\newcommand{\aas}{a.a.s.\@\xspace}
\makeatother
\newcommand{\ndense}{almost-complete\xspace}

\newcommand{\F}{\mathbb{F}}

\newcommand{\affineDAG}{\Pi}
\newcommand{\edgeProbResLin}{\alpha}
\newcommand{\betabound}{1/300}
\newcommand{\rbound}{1/100(1-\edgeProbResLin)}

\usepackage{ifthen}

\newboolean{conferenceversion}
\setboolean{conferenceversion}{false}

\ifthenelse{\boolean{conferenceversion}}
{}{}

\newcommand{\myparagraph}[1]{\paragraph{#1}}

\begin{document}
\title{
Average-Case Hardness of Binary-Encoded Clique \\in Proof and Communication Complexity}
\author{Susanna F.\ de Rezende \\
  Lund University
  \and
  David Engström  \\
  Lund University
  \and  
  Yassine Ghannane \\
  University of Copenhagen \\and Lund University
  \and
  Duri Andrea Janett \\
  University of Copenhagen \\
  and Lund University
  \and
  Artur Riazanov\\
  EPFL}
\date{\today}

\maketitle

\begin{abstract}
    We study the average-case hardness of establishing that a graph does not have a large clique in both proof and communication complexity.
    We show exponential lower bounds on the length of cutting planes and 
    bounded-depth resolution over parities 
    refutations 
    of the binary encoding of clique formulas on randomly sampled dense graphs. 
    Moreover, we show that  
    the randomized communication complexity of 
    finding a falsified clause in these formulas 
    is polynomial.
\end{abstract}

\thispagestyle{empty}
\newpage
\ifthenelse{\boolean{conferenceversion}}
{}{\tableofcontents}

\newpage

\section{Introduction}
\label{sec:introduction}

The \kClique problem asks whether a given graph $G$ contains 
a complete subgraph of size $k$.
This is one of the classical $\NP$-hard problems \cite{Karp1972}.
Under the \emph{Exponential Time Hypothesis},
solving \kClique on an $\graphSize$-vertex graph requires time $\graphSize^{\Omega(k)}$ \cite{CHKX04approx}, nearly matching the $O(\graphSize^k)$ upper bound that can be obtained by simply checking all $\binom{\graphSize}{k}$ possible subgraphs.
Furthermore, it is $\NP$-hard to approximate the maximum size of a clique even within a factor $\graphSize^{1-\varepsilon}$~\cite{haastad1999clique, Zuckerman07}.

Showing unconditional super-polynomial lower bounds for \kClique seems beyond the reach of current techniques, as it would immediately imply $\mathbf{P} \neq \NP$. Thus, significant effort has been made to provide unconditional \emph{evidence} for the hardness of \kClique. The most direct evidence is hardness in restricted computational models. 
In circuit complexity, 
\kClique was the first 
problem proven to require exponential-size \emph{monotone Boolean circuits} \cite{Razborov1985}. 
In a later breakthrough, Rossman \cite{Rossman08ConstantDepth, Rossman14Monotone} showed 
that the \kClique problem is hard for monotone and bounded-depth circuits even \emph{on average}. 
\subsection{Clique in Proof Complexity}
Another type of evidence for hardness comes from \emph{propositional proof complexity},
the study of 
the complexity of certifying, in various formal proof systems, that a given Boolean formula is a tautology or that it is unsatisfiable.

Proof complexity provides a perspective to analyze several algorithmic approaches to the \kClique problem:
If an algorithm rejects a \kClique instance, 
the trace of its execution can be viewed as a certificate 
of the claim that the given graph \emph{does not} contain a $k$-clique.
This claim can be encoded as a CNF formula asserting that the graph \emph{does} have a $k$-clique, and in this paper we  study certificates of unsatisfiability, also called \emph{refutations}, of this formula.  
Many classes of algorithms can be captured by propositional proof systems in this way: the core reasoning of state-of-the-art 
SAT-solvers is captured by \emph{resolution} \cite{BKS04}; Gr\"oebner basis algorithms by \emph{polynomial calculus} \cite{CEI96}; the cutting planes methods in integer linear optimization by the \emph{cutting planes} proof system~\cite{CCT87ComplexityCP};
and the Lassere hierarchy in semi-definite programming~\cite{Lasserre01} by the \emph{sum-of-squares} proof system. 

Assuming the \emph{Nondeterministic Exponential Time Hypothesis} \cite{CIMPR16}, in \emph{every} proof system the shortest refutation of \kClique for \emph{some} graph has length at least $\graphSize^{\Omega(k)}$.
On the other hand, unconditional lower bounds are only known for very restricted proof systems. 
Two encodings of the clique problem into CNF formulas are studied in the literature.
The \emph{unary encoding} uses $\graphSize$ propositional variables to encode the $i$th vertex of the purported $k$-clique, for every $i$,
while the more succinct \emph{binary encoding} uses only $\log \graphSize$ variables to encode the same information in binary. 

The following results are known about the unary encoding. 
For \emph{treelike} resolution, there are $\graphSize^{\Omega(k)}$ length lower bounds even 
on average 
with respect to the \erdosrenyi graph distribution~\cite{BGL13ParameterizedDPLL,Lauria18Cliques}. 
The current best known results are analogous average-case bounds for \emph{regular} resolution~\cite{ABdRLNR21Clique}, and Sherali--Adams with bounded coefficients~\cite{dRPR23UnarySA}.
For general resolution, in the parameter range $\graphSize^{5/6}\ll k \leq \graphSize/3$, an $\exp (\graphSize^{\Omega(1)})$ average-case lower bound for dense graphs is known
\cite{BIS07IndependentSets}, and when $k<\graphSize^{1/3}$, an average-case lower bound $2^{k^{1-\varepsilon}}$ holds~\cite{Pang21}.
Improving this to an optimal $\graphSize^{\Omega(k)}$ lower bounds is a long-standing open problem, see, e.g.,~\cite{BGLR12Parametrized,BN21ProofCplxSAT}, even for worst-case instances. 
In particular, all known techniques for proving resolution lower bounds seem to fail for this task~\cite{ABdRLNR21Clique}.

Weak proof systems can, however, be very sensitive to the choice of encoding, as was shown in, e.g.,~\cite{DGGM24Binary}.
This is also the case for the binary encoding of clique, where strong lower bounds are known for resolution:
There is an $\graphSize^{\Omega(k)}$ average-case lower bound for resolution~\cite{LPRT17ComplexityRamsey}, which was further extended to $s$-DNF resolution, for $s=o(\log \log n)$, at the cost of some loss in the exponent of the bound depending on $s$~\cite{DGGM24Binary}.

\subsection{Our contributions} 

We prove new lower bounds for $k$-\textsf{Clique} in three settings. Our first result is that
the binary encoding of $k$-\textsf{Clique} requires exponential proof length in cutting planes.
Prior to our work, nothing was known about the hardness of \kClique for cutting planes, even in the treelike setting.

\begin{theorem}[Informal]\label{thm:CP-informal}
    For $\graphSize$-vertex graphs $G$ sampled from the \erdosrenyi distribution with appropriate edge density, any cutting planes refutation of the binary-encoded clique formula must have length at least $2^{\graphSize^{\Omega(1)}}$.
\end{theorem}

We consider next the frontier proof system \emph{resolution over parities} ($\ResParity$)~\cite{IS20ResLin},
which recently received a lot of attention~
\cite{BhattacharyaCD24,AI25lifting,EI25amortized,BC25exponential,BI25bphp}, particularly following the breakthrough of~\cite{EGI24regular}.
Lower bounds are only known for restricted versions of this system, and in particular,
nothing is known about the hardness of refuting $k$-\textsf{Clique} in $\ResLin$.
We prove the following length lower bound on $\ResLin$ proofs of bounded depth.
\begin{theorem}[Informal]\label{thm:ResLin-informal}
    For $\graphSize$-vertex graphs $G$ sampled from the \erdosrenyi distribution with appropriate edge density,
    any $\ResLin$ refutation of the binary-encoded clique formula in depth $m^{1.5-\varepsilon}$, where $m$ is the number of variables in the formula, must have length at least~$2^{\graphSize^{\Omega(1)}}$.
\end{theorem}

Lastly, we study the communication version of the problem. 
Two parties are both given partial information about the members of a purported $k$-clique in a graph $G$.
Again, $G$ does not actually have a $k$-clique, and the parties must find a missing edge 
in
the purported clique, while communicating as few bits as possible.
This corresponds to finding a falsified clause in the binary-encoded clique formula, if the partial information consists of the bits encoding the purported clique members.
We prove the following lower bound
in the setting where the parties have access to shared randomness.

\begin{theorem}[Informal]\label{thm:Comm-informal}
    For $\graphSize$-vertex graphs $G$ sampled from the \erdosrenyi distribution with appropriate edge density, the randomized communication cost of finding a falsified clause in the binary-encoded clique formula is $\graphSize^{\Omega(1)}$.
\end{theorem}
Randomized communication complexity remains the most versatile tool for obtaining lower bounds for \emph{treelike} cutting planes. Most of the known lower bounds are obtained via this connection~\cite{IPU94Treelike,BPS07LS,HuynhN12,GP18CommunicationLowerBounds,IR21}. 
Treelike cutting planes is a natural next candidate for proving clique lower bounds in the \emph{unary} encoding.

There is, however, a growing evidence that the randomized communication model is too strong to give strong lower bounds for \kClique. If, for example, the density of \erdosrenyi graph is constant, then Alice and Bob can randomly sample a pair of indices of their nodes and verify whether or not the corresponding pair is connected with an edge. Thus, we can only hope for a non-trivial lower bound for the dense setting of \kClique. 
Another weakness of the communication approach to \kClique lower bounds was established by Jukna~\cite{Jukna12}, who observed that for a \emph{bipartite} version of the unary \kClique, the corresponding problem has a $O(\log n)$-cost \emph{deterministic} communication protocol.

Given these drawbacks, we view \cref{thm:Comm-informal} as an indication of some hope for the communication approach to \kClique lower bounds: at least in the dense binary setting the barriers above do not apply. Moreover, \cref{thm:Comm-informal} (qualitatively) generalizes the main theorem in \cite{YZ24Comm} with a much simpler proof.

\subsection{Discussion}

Our contributions can be viewed
as matching the lower bounds known for the \emph{weak} version of  \emph{binary pigeonhole principle} (BPHP) 
with bounds for the binary encoding of \kClique. 
In fact, all  three bounds are obtained by adapting 
the techniques previously used to prove lower bounds for the weak binary pigeonhole principle~\cite{BeameW2025multiparty,BI25bphp}, or by adapting extensions of such techniques~\cite{GoosGJL2025quantum,YZ24Comm}. 
Is there a more general principle behind our results, i.e., is there an explicit reduction from BPHP to average-case \kClique?
We know such a reduction exists in the \emph{worst case}---indeed, the BPHP is a special case of the binary-encoded \kClique---but in the average case such reductions are not known.

One concrete direction 
would be to prove 

lower
bounds on \kClique in the treelike version of $\Th(2)$, the generalization of cutting planes refutations to degree-$2$ inequalities. 
Only the \emph{strong} BPHP is known to be hard for this system \cite{IR21}, and the proof of this fact uses global symmetries of the formula. Therefore, it is unclear how to apply the technique to \kClique for an \erdosrenyi random graph, since the latter is unlikely to have symmetries. 
Does it fail because of the particular technique, or is the strong version of BPHP crucially easier than the weak one? 
\begin{problem}
    Suppose weak BPHP requires exponential tree-like $\Th(2)$ refutation. Show that the binary encoding of $k$-\textsf{Clique} does as well. 
\end{problem}
Another interesting direction to extend our lower bounds is to the \emph{Ramsey principles}~\cite{LPRT17ComplexityRamsey}.  
A graph is \emph{$c$-Ramsey} if it has neither a clique nor an independent set of size $c\log n$. 
The Ramsey principle formulas defined in~\cite{LPRT17ComplexityRamsey} encodes in binary
the claim that a given graph is $c$-Ramsey. For resolution, an asymptotically optimal $n^{\Omega (\log n)}$ length lower bound was proven in \cite{LPRT17ComplexityRamsey}.
We ask the same question for cutting planes and $\ResLin$.
\begin{problem}
    Prove that the Ramsey principles are hard for cutting planes or bounded-depth $\ResLin$.
\end{problem}
In particular, matching the bounds
known for resolution would require 
improving our lower bounds and 
extending them to randomly sampled graphs for a broader range of parameters,
specifically to graphs which are not as dense.

As discussed above, for sparse graphs there is always a \emph{randomized} communication protocol solving \kClique in binary or unary encoding. However, we can weaken the communication model so it still captures tree-like cutting planes. One such model is a \emph{deterministic} communication protocol with a \textsc{Greater Than} oracle, i.e., the players can compare real numbers at unit cost. 

This model is shown to be significantly weaker than randomized communication \cite{CLV19,CHHNPS25,GoosHR25}. Can we still solve \kClique over a sparse graph?
\begin{problem}
    Suppose that $\Pi$ is a deterministic communication protocol with a \textsc{Greater Than} oracle. Can $\Pi$ solve \kClique for \erdosrenyi graphs of constant sparsity with $O(\log n)$ cost?
\end{problem}

\subsection{Techniques}
We give a brief technical overview of our proofs. Notably, all of our techniques are to some extent based on previous work proving lower bounds on BPHP. 

\myparagraph{Cutting Planes.}
\autoref{thm:CP-informal} is proved via a \emph{bottleneck counting} argument \cite{HAKEN1999326}. The approach is based on
an adaptation of this framework to cutting planes due to Sokolov~\cite{Sokolov24RandomCNFs}, that was also subsequently refined by Beame and Whitmeyer~\cite{BeameW2025multiparty} to prove cutting planes $\mathrm{BPHP}$ lower bounds. It leverages a connection between cutting planes proofs and triangle-DAG protocols, saying that small cutting planes refutations of some formula imply small protocols for a related search problem. Informally, the main idea is to then define a partial function $\mu$ that maps inputs of the search problem to nodes of the protocol. Intuitively, $\mu$ measures progress made by the protocol at some step. By giving a lower bound on the size of the domain of $\mu$ and an upper bound on the maximum number of inputs mapped to each node, we establish a size lower bound on the protocol and thus on cutting planes refutations of our starting formula. The twist of our approach is to show that a certain combinatorial property of graphs related to the density of neighborhoods is sufficient to construct such a partial map $\mu$. In particular, dense random graphs asymptotically almost surely exhibit such properties.

\myparagraph{Bounded depth $\ResParity$.}
Recent progress has led to a BPHP lower bound for depth-$m^{1.5-\varepsilon}$ $\ResLin$~\cite{BI25bphp}, 
based on a notion of \emph{closure} for a set of linear forms~\cite{EGI24regular}, and 
the \emph{random walks with restarts} framework~\cite{AI25lifting}.
This framework works by performing a random walk on the proof in a top-down manner. 
If one can bound below the probability that this walk proceeds for many steps without reaching a falsified clause, and assuming that the given refutation is small, some node in the proof must be visited by many random walks, and thus the rank of the linear system at that node must be low.
Using the notion of closure, one can then fix only few variables to satisfy that system, which allows repeating this procedure, walking further down the proof, and eventually establishing a depth lower bound. 

Adapting the framework to binary-encoded clique formulas, finding a falsified clause is the same as identifying a non-edge in the underlying graph. 
Our random walk then mimics a uniformly drawn input while maintaining a partial assignment of the clique members to vertices without non-edges between them.
We ensure this by always staying within the common neighborhood of previously assigned vertices,
and find that the walk can proceed for many steps with sufficient probability. 
Additionally, we use concentrations bounds on the number of common neighbors of a set of vertices to show that, in the right parameter regimes,
after fixing variables, the remaining instance 
is similar enough to a smaller instance of the original problem, thus allowing us to restart the random walk. 
\myparagraph{Randomized communication.} 
\cref{thm:Comm-informal} (qualitatively) generalizes the BPHP lower bound of Yang and Zhang \cite{YZ24Comm}. They used query-to-communication lifting machinery from \cite{GPW20} to essentially reduce the communication lower bound to a corresponding \emph{decision tree} lower bound. The proof in \cite{YZ24Comm} is quite involved due to a white-box use of the lifting machinery.
Later, \cite{GoosGJL2025quantum} proved a structural theorem that allows a black-box proof of the main result in \cite{YZ24Comm}, and this 
theorem was recently extended and used in \cite{RiazanovSSY2025searching} to get a lower bound for the clause-search problem for random $O(\log n)$-CNFs. 
We show how to apply this framework to the clique problem.

\ifthenelse{\boolean{conferenceversion}}{
\subsection{Outline of This Paper}
The rest of this paper is organized as follows. In \autoref{sec:preliminaries} we present some preliminaries and in \autoref{sec:density-properties} we prove that random graphs satisfy some combinatorial properties. In \autoref{sec:binCliqueDagCP}
in \autoref{sec:binCliqueResLin}
in \autoref{sec:binCliqueRandComm}
}{}
\section{Preliminaries}
\label{sec:preliminaries}

We use the symbol $\sqcup$ to denote a disjoint union of sets. We write $\{1, \dots, n\}$ as $[n]$, where $n \in \mathbb{N}^+$.
Let  
$\bm x \sim A$ denote $\bm x$ uniformly sampled from a set $A$. 
We say that an event occurs \emph{asymptotically almost surely (a.a.s)} if it occurs with probability $1-o(1)$ as $n\rightarrow \infty$. We will also need the following special case of the Chernoff bound:
\begin{theorem}[Multiplicative Chernoff bounds]
Let $\bm X_1,\dots,\bm X_n$ be independent $\{0,1\}$-valued random variables, and let
$\bm X = \sum_{i=1}^n \bm X_i$, $ 
\mu = \mathbb{E}[\bm X].
$
Then, for $0\leq \delta\leq 1$,
\begin{align*}
\Pr\left[\bm X \geq (1+\delta)\mu\right] \leq \exp\left(-\frac{\mu \delta^2}{3}\right),
\\
\Pr\left[\bm X \leq (1-\delta)\mu\right] \leq \exp\left(-\frac{\mu \delta^2}{2}\right).
\end{align*}
\end{theorem}

\subsection{Graph Theory}
Let $G$ be a graph over $k$ blocks $V(G)=V_1 \sqcup \dots \sqcup V_k$, with $\lvert V_i \rvert = n$ for all $i\in [k]$.  
In everything that follows, let us assume, without loss of generality, that $n$ is a power of 2. 
For $u\in V(G)$, let $B(u)\in [k]$ denote the index of the block that contains $u$, and for every set of vertices $U \subseteq V(G)$, we define $B(U) \coloneqq \{B(u)\;\mid \; u \in U\}$. For a subset $U\subseteq V(G)$ we use $N^\cap(U,i)$ to denote the \emph{common neighborhood} of $U$ at block $i$. For a set of blocks $I \subset [k]$, we then use $N^\cap(U,I) = \bigcup_{i \in I} N^\cap(U,i)$. For all $i \in [k]$, every $u \in V_i$ can be identified with an integer in $[n]$, and we denote its binary representation by $\mathsf{bin}(u) =u_1\ldots u_{\log(n)}$.

 The natural distribution for which \kClique is conjectured to be hard is an \erdosrenyi distribution. The $p$-biased \erdosrenyi distribution $\calG(n,p)$ is defined by sampling an $n$-vertex graph with each of the $\binom{n}{2}$ edges present with probability $p$ independently of the other edges. The parameter $p$ that conjecturally makes finding a $k$-clique in $\bm{G} \sim \calG(n,p)$ hard is at the threshold $p'$ where for a lower $p<p'$ the graph $\bm G$ does not have a $k$-clique \aas, and for a higher $p>p'$ it does have a $k$-clique \aas. By looking at the expected number of $k$-cliques in $\bm G$, one can conclude that the threshold value is $\Theta(n^{-2/(k-1)})$. 

 We also define the \emph{$k$-partite} version of the same distribution: For a fixed partition of the nodes, $\calG(n, p, k)$ is sampled from by first sampling $\bm G \sim \calG(nk, p)$ and then intersecting $\bm G$ with a complete $k$-partite graph with $n$ nodes in each part. 

\subsection{Proof Complexity}
We next recall some basic notions from proof complexity; see, e.g.,~\cite{Krajicek19ProofComplexity,BN21ProofCplxSAT} for a more thorough exposition.  
A Boolean variable $x$ or its negation $\overline{x}$ is called a \emph{literal},
and a disjunction of literals over pairwise disjoint variables $C= \ell_1 \lor \cdots \lor \ell_k$ is called a \emph{clause}.
A \emph{CNF formula} is a conjunction of clauses $F = C_1 \land \cdots \land C_m$. 
We sometimes call the clauses of $F$ \emph{axioms}, and denote the set of variables occurring in $F$ as $\textrm{Vars}(F)$.

\myparagraph{Cutting Planes.}
The \emph{cutting planes} proof system~\cite{CCT87ComplexityCP} operates on systems of linear inequalities. In its syntactic form, it derives new constraints from previously derived ones by \emph{linear combination} and \emph{division with rounding}. Since our lower bounds also apply to it, we define the stronger \emph{semantic} version of the proof system, which subsumes all usual cutting planes derivation rules. 

Let $Ax\geq b$ be a system of linear inequalities, where $A\in \Z^{m\times n}$ and $b\in \Z^m$.
A \emph{semantic cutting planes refutation} of $Ax\geq b$ is a sequence of linear inequalities 
$\{d_i x\geq c_i\}_{i\in [\prflength]}$, 
where $d_i \in \R^n$ and $c_i\in \R$, such that the following hold: The final inequality is the trivially false $0\geq 1$, and
for every $i \in [\prflength]$, the inequality $d_i x\geq c_i$ is part of the linear system $Ax\geq b$, or there are $j,k<i$ such that $d_i x\geq c_i$ follows from $d_j x\geq c_j$ and $d_k x\geq c_k$ by \emph{semantic deduction}. 
That is, from inequalities $d_j x\geq c_j$ and $d_k x\geq c_k$, one can derive any inequality $d_i x\geq c_i$ for which it holds that any $x \in \{0,1\}^n$ satisfying both $d_j x\geq c_j$ and $d_k x\geq c_k$ satisfies $d_i x\geq c_i$.
The \emph{length} of a semantic cutting planes refutation is the number $\prflength$ of inequalities appearing in the sequence.

The following translation enables us to consider cutting planes refutations of unsatisfiable CNF formulas.  
For every propositional variable $x$, we add the constraints $0\leq x$ and $x\leq 1$. For a clause $C=\bigvee_{i\in I}x_i\lor\bigvee_{j\in J} \overline{x_j}$, we add the inequality $\sum_{i\in I} x_i + \sum_{j\in J} (1-x_j) \geq 1$.

\myparagraph{Resolution over Parities.} 
The \emph{resolution over parities} proof system~\cite{IS20ResLin}, which we refer to as $\ResLin$, operates on linear equations over $\F_2$.
Let $x_1, \dots, x_n$ be variables taking values in $\F_2$. A \emph{linear form} is a polynomial $\sum_{i=1}^n a_ix_i$, where all $a_i\in \F_2$. Taking a linear form $f$ and $a\in \F_2$, we get a \emph{linear equality} $f=a$. 
A disjunction of linear equations $C=\bigvee_{i=1}^m (f_i=a_i)$ is called a \emph{linear clause}. Sometimes, it is convenient to view a linear clause as a negation of a linear system, i.e., $C=\lnot \bigwedge_{i=1}^m (f_i = a_i+1)$. 
Note that a propositional clause $C=\bigvee_{i\in I}x_i\lor\bigvee_{j\in J}x_j$ is a special case of a linear clause: we can write $C=\bigvee_{i\in I} (x_i=1) \lor \bigvee_{j\in J}(x_j=0)$.

Given a CNF formula $F$, a \emph{$\ResLin$ refutation 
of $F$} is a sequence of linear clauses $\{C_i\}_{i\in\prflength}$ ending with the empty clause and where every clause either belongs to $F$ or was derived from two previous clauses using one of the following two derivation rules:  The \emph{resolution rule}
\begin{equation*}
  \AxiomC{$C \lor (x=1)$}
  \AxiomC{$D \lor (x=0)$}
  \BinaryInfC{$C \lor D$}
  \DisplayProof ,
\end{equation*}
or the \emph{semantic weakening rule}, which allows to derive from a linear clause $C$ any linear clause $D$ that semantically follows from it, i.e., such that any $x\in \{0,1\}^n$ satisfying $C$ also satisfies $D$.  

The \emph{length} of a $\ResLin$ refutation is the number $\prflength$ of linear clauses appearing in it. The \emph{depth} $d$ of a $\ResLin$ refutation is the largest number of resolution steps among all the paths from axioms to the empty clause in the refutation.
The proof system \emph{depth-$d$ $\ResLin$} is the subsystem of $\ResLin$ consisting only of refutation of depth at most $d$, where $d$ can depend on the number of variables $n$.

\myparagraph{Shape-DAGs.}
We will also work with the top-down definition of the proof systems defined above, namely, \emph{triangle-} and \emph{affine-DAGs} solving \emph{total search problems}. 
A \emph{total search problem} is defined as a relation $S\subseteq \mathcal{I} \times \mathcal{O}$ over finite sets of inputs $\mathcal{I}$ and outputs $\mathcal{O}$, such that for every input $x\in \mathcal{I}$ there is an output $o\in \mathcal{O}$ with $(x,o)\in S$. Let $S^{-1}(o)= \{x\in \mathcal{I} \mid (x,o) \in S \}$.

For a given unsatisfiable CNF formula $F= C_1\land \dots C_m$, we consider the \emph{falsified clause search problem $\search_{F}$}: Given an assignment $x\in \{0,1\}^n$ to the variables of $F$, output $i\in [m]$, such that the clause $C_i$ is falsified under $x$. In detail, $(x,i)\in \search_F$ if and only if $C_i(x)=0$. 

\begin{definition}[\cite{GGKS20MonotoneCircuit}]
    Let $S\subseteq \mathcal{I}\times \mathcal{O}$ be a total search problem, $\mathcal{F}\subseteq 2^{\mathcal{I}}$. We call the elements of $\mathcal{F}$ \emph{shapes}.
    A \emph{shape-DAG} ($\mathcal{F}$-DAG) solving $S$ is a rooted directed acyclic graph $D$ of fan-out at most $2$, where each node $v\in D$ is labeled by a shape $F_v\in \mathcal{F}$ such that 
    \begin{enumerate}
        \item the root $r\in D$ is labeled by $F_r=\mathcal{I}$;
        \item for a node $v$ with children $u,w$, $F_v\subseteq F_u \cup F_w$; 
        \item and for every leaf $
        \ell$, there is an output $o \in \mathcal{O}$ such that $F_{\ell}\subseteq S^{-1}(o)$.
    \end{enumerate}
    The \emph{depth} of $D$ is the longest root-to-leaf path in $D$; its \emph{length} is the number of nodes in it.
\end{definition}

Given a bipartite domain $X\times Y$, a \emph{triangle} $T\subseteq X \times Y$ is a set which can be written as $T=\{(x,y)\in X\times Y \mid a_T(x)\leq b_T(y)\}$, for some $a_T\colon X \rightarrow \R$, $b_T\colon Y \rightarrow \R$. 
A \emph{triangle-DAG} is a shape-DAG where the shapes are all triangles. To consider triangle-DAGs for the falsified clause search problem, we split the variables of $F$ into two parts and view $\{0,1\}^n$ as the corresponding product. Let $\search_F^{X,Y}\subseteq (X\times Y) \times [m]$ denote $\{((x,y),i) \mid C_i(x,y)=0\}$.
As the following proposition states, lower bounds from cutting planes follow from lower bounds for triangle-DAGs.
\begin{proposition}[\cite{Sokolov17Daglike,HP18Note}]\label{prop:CP-to-triangleDAG}
A semantic cutting planes refutation of a given CNF formula $F$ yields, for any partition of the variables $\textrm{Vars}(F)=X\sqcup Y$, a triangle-DAG solving $\search_F^{X,Y}$ of the same length.
\end{proposition}

An \emph{affine subspace-DAG} is a shape-DAG where the shapes are affine subspaces of the domain $\F_2^n$. An \emph{affine-DAG} is an affine subspace-DAG $D$, where  every non-leaf node $v\in D$ has two outgoing edges, $(v,w)$ and $(v,w')$, labeled $P_v=0$ and $P_v=1$, respectively. In addition, for the linear systems $\Psi_v,\Psi_w,\Psi_{w'}$ corresponding to the affine subspaces labeling $v,w,w'$, it holds that $\Psi_w$ is implied by $\Psi_v \land (P_v=0)$ and $\Psi_{w'}$ is implied by $\Psi_v \land (P_v=1)$. 
Similar to \autoref{prop:CP-to-triangleDAG}, lower bounds for affine-DAGs can be used to derive lower bounds for $\ResLin$, as stated formally below.
\begin{proposition}[\cite{EGI24regular}]\label{prop:ResLin-to-affineDAG}
    A $\ResLin$ refutation of a given CNF formula $F$ yields an affine-DAG solving $\search_F$ of the same length and depth.
\end{proposition}

\subsection{Binary Encoding of Clique Formula}
\label{sec:bin-encoding}
We encode the claim that a graph $G$ has a clique of size $k$. We use the so-called \emph{binary encoding}~\cite{LPRT17ComplexityRamsey,DGGM24Binary}, where for every $i\in [k]$, there are $\log \setsize{V}$ variables to point out the $i$th member of a purported $k$-clique.
For a propositional variable $x$, we use the notation $(x \neq 0)$ to denote the literal $x$, and $(x \neq 1)$ to denote the literal $\overline{x}$.

Assume that the vertices $v_i$, $v_j$ and $v$ are represented as $v_{i,1}, \dots, v_{i, \log \setsize{V}}$, 
$v_{j,1}, \dots, v_{j, \log \setsize{V}}$, and
$v_{1}, \dots, v_{\log \setsize{V}}$ in binary, respectively.
Let $x_{i,a}$, where $i\in [k]$ and $a \in [\log \setsize{V}]$, be propositional variables. The formula $\binCliqueGk$ consists of the clauses
\begin{align}
    \label{eq:binary-encoding1}
    & \bigvee_{a=1}^{\log \setsize{V}} (x_{i,a}\neq v_{i,a}) \lor  \bigvee_{a=1}^{\log \setsize{V}} (x_{j,a}\neq v_{j,a}), 
    & \qquad \forall v_i\neq v_j\in V, i\neq j \in [k]: \{v_i,v_j\}\notin E, \\
    \label{eq:binary-encoding2}
    & \bigvee_{a=1}^{\log \setsize{V}} (x_{i,a}\neq v_{a}) \lor  \bigvee_{a=1}^{\log \setsize{V}} (x_{j,a}\neq v_{a}), 
    & \qquad \forall v\in V, i\neq j \in [k].
\end{align}
The clauses \eqref{eq:binary-encoding1} are called \emph{edge axioms}, and the clauses \eqref{eq:binary-encoding2} are called \emph{functionality axioms}.
The edge axioms encode that two non-neighbors are not simultaneously chosen as clique members, while the functionality axioms ensure that all clique members are distinct.

\myparagraph{Block encoding.} 
We next consider a version of the binary clique formula, where cliques are required to have a ``block-respecting'' structure. 
That is, we partition the vertices of $G$ into $k$ \emph{blocks} of equal size $V=V_1 \sqcup \dots \sqcup V_k$.
We then encode the claim that $G$ has a $k$-clique with one vertex from every block. Such a clique is also called \emph{transversal}.

Assume that $v_i$ is represented as 
$v_{i,1}, \dots, v_{i, \log n}$ and $v_j$ is represented as 
$v_{j,1}, \dots, v_{j, \log n}$. Let $x_{i,a}$, where $i\in [k]$ and $a \in [\log n]$, be propositional variables. Then the formula $\binCliqueBlock{G}{k}$ consists of the clauses
\begin{align}
    \bigvee_{a=1}^{\log n} (x_{i,a}\neq v_{i,a}) \lor  \bigvee_{a=1}^{\log n} (x_{j,a}\neq v_{j,a}), 
    \qquad \forall v_i\in V_i, v_j \in V_j: \{v_i,v_j\}\notin E.
    \label{eq:block-encoding}
\end{align}
The clauses \eqref{eq:block-encoding} are edge axioms, and again, encode that for $(v_i,v_j)\notin E$, ``$v_i$ is not the clique member from block $i$ or $v_j$ is not the clique member from block $j$''.
Note that in the block encoding, we do not need functionality axioms, since the clique members are from distinct blocks by definition.

The following lemma shows that lower bounds for $\binCliqueBlockGk$ imply lower bounds for $\binCliqueGk$. Therefore, for the rest of the paper, we will only consider the block encoding.
\begin{lemma}[\cite{BIS07IndependentSets}]\label{lem:blockToBinary}
    Let $k\in \mathbb{N}^+$ and $G$ be a graph. 
    If there is a cutting planes (respectively, $\ResLin$) refutation of $\binCliqueGk$ of length $L$ and depth $d$, then there is a cutting planes (respectively, $\ResLin$) refutation of $\binCliqueBlockGk$ of length at most $L$ and depth at most~$d$. 
\end{lemma}
While the statement in~\cite{BIS07IndependentSets} is for the unary encoding of clique formulas and the resolution proof system, its proof can be modified to obtain the lemma above. 
Indeed, 
one can obtain $\binCliqueBlockGk$ from $\binCliqueGk$ by applying a random restriction, and both cutting planes and $\ResLin
$ are closed under random restrictions.

\subsection{Communication Complexity}
In this paper, we consider randomized (public coin)
two party
number-in-hand communication complexity. 
We quickly define the model and refer to \cite{KN97CommunicationCplx,RY20CommunicationCplx} for details.

Two players, \emph{Alice} and \emph{Bob}, 
are given inputs $x\in X$ and $y\in Y$, respectively. 
To solve a search problem $S\subseteq (X \times Y) \times \mathcal{O}$, they need to find $o\in \mathcal{O}$ such that $((x,y),o)\in S$, 
while communicating as few bits as possible. 
A \emph{deterministic communication protocol} $\protocolPi$ for $S$ is a rooted binary tree as follows. Every internal node $v\in\protocolPi$ determines who is going to speak, and the spoken bit, which is a function of $v$ and the input $x$ or $y$ ($x$ if Alice speaks, $y$ otherwise), determines which child of $v$ the computation moves to. Every leaf $\ell \in \protocolPi$ is labeled by $o_{\ell} \in \mathcal{O}$, such that $((x,y),o_{\ell})\in S$.
We define a \emph{randomized communication protocol} $\protocolPiRand$ for $S$ as a distribution over deterministic communication protocols, such that the computation is correct with constant probability.
The \emph{cost} of a deterministic protocol $\commCost{\protocolPi}$ is the maximum, over all inputs, of the number of bits communicated. 
The \emph{cost} of a randomized protocol $\commCost{\protocolPiRand}$ is the maximum cost among the deterministic protocols in the distribution.

We study the communication complexity of $\search^{X,Y}_{\binCliqueGk}$, where we split the variables of $\binCliqueGk$ such that Alice and Bob each get half of the bits from every clique member, i.e., $X=Y=(\{0,1\}^{\log \setsize{V}/2})^k$, and denote this problem by $\searchComm{G}$. 
As in the proof complexity setting, a version of \autoref{lem:blockToBinary} applies and we may prove lower bounds for the block encoding, i.e., for $\binCliqueBlockComm{G}{k}$. 

Furthermore, we will make use of the characterization of randomized communication cost in terms of distributional communication complexity~\cite{yao83}: To prove a lower bound on the randomized communication cost, it is enough to prove a lower bounds on deterministic protocols with constant error probability for \emph{any} input distribution. 
In particular, we will prove a lower bound for the uniform distribution.
\section{Density Properties of Random Graphs}
\label{sec:density-properties}

In the subsequent sections, we prove hardness results for graphs satisfying some combinatorial property related to the density of their neighbor sets.  
In this section, we define the properties we consider and \ifthenelse{\boolean{conferenceversion}}
{verify}{prove} that random graphs satisfy them.

\begin{definition}\label{def:neighbor-dense}

Given $s,k,n > 0$, we say a graph $G$ over $k$ blocks of size $n$ is \emph{$s$-\ndense} if for every $i,j\in [k]$, $i\neq j$,  
and $x_i,y_i, x_j \in \Sigma$, where $\Sigma = \{0,1\}^{\log (n)/2}$, there are at most $s$ different $y_j$ such that there is no edge between $(x_i,y_i)$ and $(x_j,y_j)$ in $G$.
\end{definition}
 \ifthenelse{\boolean{conferenceversion}}{Random}{We show that random} dense graphs \aas satisfy this property for an appropriate value of $s$. 

\begin{lemma}
    \label{lemma:probability-concentrates}
    For $p \in [0,1]$, if $\random{G} \sim  \calG(n,p,k)$, then, \aas, $\random{G}$ is $\max(2\sqrt{n}(1-p),9e^2\log(kn))$-\ndense.
\end{lemma}

\ifthenelse{\boolean{conferenceversion}}
{}{
\begin{proof}
    For fixed 
    $i,j\in [k]$
    and fixed $x_i$, $y_i$, and $x_j$, let $\random{Z}$ be the random variable denoting number of $y_j$ such that there is no edge between $(x_i,y_i)$ and $(x_j,y_j)$. First, suppose that $p\leq1-\frac{9\log(nk)}{\sqrt n}$. Chernoff's inequality implies that 
    \begin{equation*}
        \Pr[\random{Z}\geq2\sqrt{n}(1-p)] \leq \exp\left ( -\frac{\sqrt{n}(1-p)}{3}\right)\,. 
    \end{equation*}

    Taking a union bound over every choice of $i,j$ and of $x_i,y_i$ and $x_j$,
    we find that the probability of a graph not satisfying the desired property is bounded by 
    \begin{align*}
        \binom{k}{2}\sqrt{n}^3 \cdot 
        \Pr[\random{Z}\geq 
        2\sqrt{n}(1-p)] &\leq k^2 n^{3/2} \cdot \exp\left (- \frac{\sqrt{n}(1-p)}{3}\right)
        \\
        &\leq k^2n^{3/2}\cdot \exp(-3\log(nk))\ll \frac{1}{nk}\,. 
    \end{align*}
    Now, we consider the case where $p > 1-\frac{9\log(nk)}{\sqrt n}$. Then a naive union bound gives
    \begin{align*}
         \binom{k}{2}\sqrt{n}^3 \cdot 
         \Pr[\random{Z}\geq 
         9e^2\log(kn)] &\leq 
         \binom{k}{2}\sqrt{n}^3 
         \binom{\sqrt n}{9e^2\log(kn)}
         (1-p)^{9e^2\log(kn)}\\
         &\leq \binom{k}{2}\sqrt{n}^3 
         \left(\frac{\sqrt n}
         {9e\log(kn)}\right)^{9e^2\log(kn)}
         \left(\frac{9\log(kn)}
         {\sqrt n}\right)^{9e^2\log(kn)}\\
         &\leq k^2 \sqrt{n}^3 \frac{1}{(kn)^{9e^2}}\,. 
    \end{align*}
    This completes the proof of the lemma.
\end{proof}
}

The following property, which appeared in \cite{dRPR23UnarySA}, is about the common neighborhood of small subsets of vertices. Similar definitions have also appeared in~\cite{BIS07IndependentSets, BGL13ParameterizedDPLL, ABdRLNR21Clique}.
\begin{definition}\label{def:common-neighbor}
    Given $\alpha,\beta,R, k,n > 0$, 
    we say a graph $G$ over $k$ blocks of size $n$ has \emph{$(\alpha,\beta,\CNbound)$-bounded} common neighborhoods if every $S\subseteq V(G)$ of size at most $\CNbound$ and every block $i \in [k]\setminus B(S)$ satisfies
    $$N^\cap(S,i)\in (1\pm \beta)\alpha^{|S|}n\,.$$
\end{definition}

For random graphs, this property expresses that the size of common neighborhoods of small sets behaves approximately as expected.

\begin{lemma}
    \label{lemma:neigbor-concentrates}
    For $0\leq\delta <1$, $p\geq 1- n^{-\delta}$, $\beta \geq \sqrt{\frac{18\log(kn)}{n^{1-\delta}}}$, if $\random{G} \sim  \calG(n,p,k)$, then, \aas, $\random{G}$ has $(p,\beta,n^\delta)$-bounded common neighbors.
\end{lemma}

\ifthenelse{\boolean{conferenceversion}}
{}{
\begin{proof}
For fixed 
$S \subseteq  V(\random{G})$
and $i \in [k]\setminus B(S)$, let $N_{\random{G}}^\cap(S,i)$ be the random variable which is equal to the number of vertices of block $i$ lying in the common neighborhood of $S$. By Chernoff's inequality,
\begin{align*}
    \Pr_{\random{G}}[\big\lvert\lvert {N_{\random{G}}^\cap(S,i)}\rvert - p^{|S|}n\big \rvert \geq \beta p^{|S|}n]
    &\leq 2\cdot \exp\left ( -\frac{\beta^2p^{|S|} \cdot n}{3}\right)\\
    & \leq 2\cdot \exp\left ( -\frac{\beta^2(1-|S|/n^\delta) \cdot n}{3}\right)\,,
\end{align*}
where we use that $p^{|S|} \geq 1-|S|n^{-\delta}$ for the second inequality. Then, by a union bound over all such sets $S$ and blocks $i$, we conclude that $\random{G}$ does not verify the property with probability at most
\begin{align*}
    2k\cdot \sum_{j=1}^{n^\delta} 
    \binom{kn}{j}\cdot 
    \exp\left ( 
    -\frac{\beta^2(1-j/n^\delta) \cdot n}{3}
    \right) &\leq 2k\cdot \sum_{j=1}^{n^\delta} (kn)^j\cdot 
    \exp\left ( -\frac{\beta^2 n}{6}\right)\\
    & \leq 4k \cdot (kn)^{n^\delta}
    \exp\left( -3n^{\delta} \log (kn) \right) \\
    & \leq 
    \exp\left( 2n^{\delta} \log (kn)-3n^{\delta} \log (kn) \right) \\
    &\leq \exp(-n^{\delta} \log (kn)) . \qedhere
\end{align*}
\end{proof}
}

\section{Lower Bound for Cutting Planes} \label{sec:binCliqueDagCP}

\newcommand{\protocol}{\Pi}
\newcommand{\bw}{\mathsf{bw}}

This section is devoted to proving our first result, stated formally below.
\begin{theorem}
    For any  
    integers $n$ and $k$, and for any real $p \in [0,1]$
    if $\bm G$ is a graph sampled from $\calG(n,p,k)$ then \aas semantic cutting planes requires length 
    \begin{equation*}
        \exp\left({\Omega\left({\min\left(
    \frac{1}{\sqrt{1-p}},
    \,
    \frac{n^{1/4}}{\log(kn)}\right)}\right)}\right)
    \end{equation*}
    to refute $\binCliqueBlockGk$.%
    \label{theorem:cutting-planes-bound}
\end{theorem}
We quickly verify that \autoref{thm:CP-informal} follows from the theorem stated above. For $p=1-n^{-1/3}$ and $k = n$, we have $\graphSize=nk=n^2$, and the bound in \autoref{theorem:cutting-planes-bound} becomes $2^{\Omega(\graphSize^{1/12})}$, 
   and it applies to $\binCliqueGk$ by \autoref{lem:blockToBinary}. Also note that at this edge density, and with $k=n$, the formula $\binCliqueGk$ is \aas unsatisfiable, meaning the bound is non-trivial.

By \autoref{prop:CP-to-triangleDAG} and \autoref{lemma:probability-concentrates}, we obtain the above theorem by proving
a size lower bound on triangle-DAG protocols solving the corresponding search problem, for the variable partition introduced in the preliminaries and over \ndense graphs.  
\comment{move to prelims :\\
Let $G$ be a graph over $k$ blocks $V(G)=V_1 \sqcup \dots \sqcup V_k$, with $\lvert V_i \rvert = n$ for all $i\in [k]$. Without loss of generality, let us assume that $n$ is a power of 2. 
For $u\in V(G)$, let $B(u)\in [k]$ denote the index of the block that contains $u$, and for every set of vertices $U \subseteq V(G)$, we define $B(U) \coloneqq \{B(u)\;\mid \; u \in U\}$. For all $i \in [k]$, every $u \in V_i$ can be identified with an integer in $[n]$, and we denote its binary representation by $\mathsf{bin}(u) =u_1\ldots u_{\log(n)}$.\\
Let also $X = Y =\Sigma^k$, where $\Sigma=\{0,1\}^{ \log(n)/2}$, such that for all $x = (x_i)_{1\leq i\leq k}\in X$, $y = (y_i)_{1\leq i\leq k}\in Y$,  $(x_i,y_i)$ uniquely identifies a vertex in the $i$-th block. This induces a partition $V_X \sqcup V_Y$ on the variables of $\binCliqueGk$, where the $V_X$ corresponds to the first $ \log(n)/2 $ variables per column, i.e $V_X = \{x_{i,j}\, \mid \, i\in [k],j\in[\log(n)/2]\}$, and $V_Y$ is the latter half, i.e $V_y = \{x_{i,j}\, \mid \, i\in [k], j\in[\log(n)/2+1,\log(n)]\}$. The lower bound for triangle-DAG protocols for this partition of the variables is then the following. 
}
\begin{theorem}
    Let $s,k,n>0$, $q = n^{1/4}/\sqrt{8s}$ and $G$ over $k$ blocks of size $n$ and $s$-\ndense then triangle-DAG protocols solving $\search_{\binCliqueGk}^{X,Y}$ require size $2^{q-2}$.
    \label{theorem:dense-dag-bound}
\end{theorem}

To see that \autoref{theorem:cutting-planes-bound} follows from \autoref{theorem:dense-dag-bound} recall that, by     \autoref{lemma:probability-concentrates}, $\bm G\sim \calG(n,p,k)$ is, \aas, $s$-\ndense for $s=\max(2\sqrt{n}(1-p),9e^2\log(kn))$. Applying \autoref{theorem:dense-dag-bound} to such an $s$-\ndense graph we have that 
\begin{equation*}
    q = \frac{n^{1/4}}{\sqrt{8s}} = \Omega\left(\frac{n^{1/4}}{\sqrt{\max(\sqrt{n}(1-p),\log(nk))}}\right) = \Omega\left({\min\left(
    \frac{1}{\sqrt{1-p}},
    \,
    \frac{n^{1/4}}{\log(kn)}\right)}\right)
\end{equation*}
from which \autoref{theorem:cutting-planes-bound} follows.
Before we present the proof of this theorem, let us establish some additional notation.
Let $G$ be a graph over $k$ blocks each of size $n$ with $V(G)=V_1 \sqcup \dots \sqcup V_k$. Recall that for $X = Y =\Sigma^k$, where $\Sigma=\{0,1\}^{ \log(n)/2}$, every $x = (x_i)_{1\leq i\leq k}\in X$, $y = (y_i)_{1\leq i\leq k}\in Y$,  $(x_i,y_i)$ uniquely identifies a vertex in the $i$-th block. 
Given $u\in V_i$ and $v \in V_j$, where $i\neq j$, we define $R_{u,v} \subseteq X\times Y$ to be the rectangle 
consisting of inputs $(x,y) \in X\times Y$ verifying $ (x_i,y_i)=u$ and $(x_j,y_j)=v$, that is,
$R_{u,v} \coloneqq \{x \in X \mid x_i = u_1\ldots u_{\log(n)/2}\,\land \, x_j = v_1\ldots v_{\log(n)/2} \} \times  \{y \in Y \mid y_i = u_{\log(n)/2+1}\ldots u_{\log(n)}\, \land \, y_j = v_{\log(n)/2+1}\ldots v_{\log(n)} \}$.
Let $\mathcal{R}$ be the set of all rectangles $R_{u,v}$ where $u$ and $v$ are non-adjacent. Additionally, for any $R \in \mathcal{R}$, we define $B(R)$ to be the set of vertices $\{u,v\}$ such that $R = R_{u,v}$. Note that $\mathcal{R}$ corresponds to the collection of sets of inputs corresponding to pre-images of solutions of $\search_{\binCliqueGk}^{X,Y}$.

For a triangle $T \subseteq X\times Y$ and an input $x \in X$, let $T^x$ denote the slice $(\{x\}\times Y)\cap T$. We define the \emph{block width} \cite{Sokolov24RandomCNFs, BeameW2025multiparty} of $x$ in $T$, denoted $\bw(T,x)$, to be the minimum number of unique blocks mentioned in a covering of $T^x$ by rectangles from $\mathcal{R}$, i.e 
\begin{equation*}
\bw(T,x) = \min_{\substack{S \subseteq \mathcal{R}\\T^x \subseteq \bigcup_{R \in S}S}}\lvert\bigcup_{R \in S} B(R) \rvert \,. 
\end{equation*}
When $T$ is evident from context, we simply call this the block width of $x$. 

The proof of the main theorem is based on a bottleneck counting argument following a framework that appeared in previous work \cite{HAKEN1999326, Sokolov24RandomCNFs, BeameW2025multiparty}, adapted to our setting. We construct a partial function $\mu\colon X\sqcup Y\to V(\protocol)$, defined in \autoref{alg:1} for which we prove two properties.
First, we show that a substantial fraction of all inputs must be assigned by $\mu$ to some node.
We then prove that no single node can have many inputs assigned to it and conclude that
there must be many nodes in the protocol. 

Informally, for a threshold $\threshold$, the map $\mu$ is constructed by traversing the nodes of the triangle-DAG in topological order from the leaves to the root and, whenever we encounter a $z\in X\sqcup Y$ such that $\bw(T,z)> \threshold$, we let $\mu$ map $z$ to that node in $\protocol$ and remove $z$ from all nodes not yet considered.

\begin{algorithm2e}[t]
\KwIn{input set $X\times Y$ and triangle-DAG $\protocol$}
\KwOut{partial map $\mu: X\sqcup Y\to V(\protocol)$}
    Let $X':= X$ and $Y' = Y$;
    \tcc{the remaining coordinates not yet assigned by $\mu$}
    \For{$u$ in a topological ordering of $V(\protocol)$, starting from the sinks}{
        \tcc{$T_u$ is the triangle at vertex $u$ in $\protocol$}
        $T'_u \coloneqq T_u \cap(X'\times Y')$\;
        \For{$x \in X'$}{
            \If{$\bw(T_u,x)>\threshold$}{
                Let $\mu(x) \coloneqq u$ and delete $x$ from $X'$\;
            }
           $T'_u \coloneqq T_u \cap(X'\times Y')$\; 
        }
        \For{$y \in Y'$}{
            \If{$\bw(T_u,y)>\threshold$}{
                Let $\mu(y) \coloneqq u$ and delete $y$ from $Y'$\;
            }
           $T'_u \coloneqq T_u \cap(X'\times Y')$\; 
        }  
    }
    \KwSty{return} $\mu$
\caption{Definition of $\mu$}
\label{alg:1}
\end{algorithm2e}

The first property of $\mu$ is straightforward\ifthenelse{\boolean{conferenceversion}}
{}{to prove}.
\begin{lemma}
\label{lemma:many x are assigned}
For $\threshold\leq n^{1/4}/\sqrt{s}$ and $G$ $s$-\ndense, then 
$\mu$ assigns at least $\lvert X \sqcup Y\rvert /4$ elements. 
\end{lemma}

\ifthenelse{\boolean{conferenceversion}}
{}{
\begin{proof}
Consider the rectangle $T'_r = X'\times Y'$ of unassigned inputs at the root. 
We prove that either $X'$ is empty and therefore $\mu$ assigns $\lvert X\rvert = \lvert X\sqcup Y\rvert /2$ elements, or $\lvert Y' \rvert \le \lvert Y\rvert/2$ and thus $\mu$ assigns at least $\lvert Y\rvert/2 = \lvert X\sqcup Y\rvert /4$ elements.

If $X'$ is empty, then we are done; otherwise, let $x$ be any element in $X'$. By construction of $\mu$,  there exists a set $S\subseteq \mathcal{R}$ that covers $\{x\}\times Y'$ and mentions at most $\threshold$ blocks. Since the graph is $s$-\ndense, and since every rectangle of $S$ is indexed by a pair of vertices that do not share an edge, 
\begin{equation}
\lvert S\rvert \leq 
\binom{\threshold}{2} 
\sqrt{n}
\cdot
s \leq 
n/2 \,,
\end{equation}
where we use that $q \leq n^{\frac{1}{4}}/\sqrt{s}$.
This implies that
\begin{align*}
    \Pr_{\random{y}\sim Y}\left[\random{y}\in Y'\right] 
    &\leq \Pr_{\random{y}\sim Y}\left[\exists R_{u,v}\in S\text{ s.t. } (x,\random{y})\in R_{u,v}\right]\\
    &\leq \lvert S\rvert /n \\
    &\leq 1/2\,,
\end{align*}
and therefore, $\lvert Y'\rvert  < \lvert Y\rvert /2$, which concludes the proof.
\end{proof}
}

By definition of a triangle-DAG, the following claim is immediate from  \autoref{alg:1}. 
\begin{claim}[{\cite{Sokolov24RandomCNFs, BeameW2025multiparty}}]
    During the execution of \autoref{alg:1}, for every $u$ in the triangle-DAG $\protocol$ and $z\in X'\cup Y'$, the block width of $z$ in $T'_u$ is at most $2\threshold$.
\end{claim}

We now state the main lemma of this section.
\begin{lemma}
    \label{lemma:few x are assigned per node}
    Let $\threshold\leq n^{1/4}/\sqrt{8s}$ and $G$ $s$-\ndense. For all $u\in V(\protocol)$, $\mu$ maps at most $\lvert X \sqcup Y\rvert \cdot 2^{-\threshold}$ elements of $X\sqcup Y$ to~$u$. 
\end{lemma}

Assuming this lemma, \autoref{theorem:dense-dag-bound} follows immediately.

\begin{proof}[Proof of \autoref{theorem:dense-dag-bound}]
     By \autoref{lemma:many x are assigned}, at least $\lvert X\sqcup Y\rvert/4$ elements from $\lvert X\sqcup Y\rvert$ are assigned to some vertex in $\protocol$ by $\mu$. On the other hand, \autoref{lemma:few x are assigned per node} bounds the number of $z \in X\sqcup Y$ assigned to any one vertex in $\protocol$. Put together, we conclude that there must be at least $\frac{\lvert X\sqcup Y \rvert/4}{\lvert X \sqcup Y\rvert 2^{-\threshold}} = 2^{q-2}$ vertices in~$\protocol$.
\end{proof}

    \begin{algorithm2e}[t]

    \KwIn{Sets $X' \subseteq X$, $Y' \subseteq Y$, and triangle $T'\subseteq X'\times Y'$ s.t.  $\bw(T',y)\leq2\threshold$ for all $y\in Y'$}
    \KwOut{A tree $\calT$ of coverings of $T'^x$ for every $x\in X'$}
    Initialize $\calT$ as a tree with a single node labeled $T'$\;
    Let $\DataSty{Leaves}:= \{T'\}$\;
    \While{$\DataSty{Leaves}\neq\emptyset$}{
        Choose some $T$ from \DataSty{Leaves} and remove it from \DataSty{Leaves}\;
        Choose a $y\in Y$ such that $\lvert T^y\rvert $ is maximized\;
        Let $R_{u_1,v_1},\dots,R_{u_\ell,v_\ell}$ be a minimal covering for $T^y$ with respect to block-width\;
        \tcc{note that $\lvert B(\{u_1,v_1,\dots,u_\ell,v_\ell\})\rvert \leq 2\threshold$}
        \For{$i\in [\ell]$}{
            Let $R:= R_{u_i,v_i} = X_R\times Y_R$\;
            Let $T_R := T\cap (X_R\times (Y\setminus Y_R))$\;
            Add an edge to $\calT$ labeled by $(u_i,v_i)$ from the node  of $\calT$ labeled $T$ to a new leaf labeled by $T_R$\;
            \If{$T_R \neq \emptyset$}{
                Add $T_R$ to \DataSty{Leaves}\;
            }
        }
    }
    \KwSty{return} $\calT$
    \caption{Definition of a tree $\calT$ of potential coverings of $T'^x$}  
    \label{alg:2}
    \end{algorithm2e}

The rest of this section is dedicated to the proof of \autoref{lemma:few x are assigned per node}. 
    Fix $u\in V(\protocol)$. Our goal is to show that the number of elements $x\in X$ assigned by $\mu$ to $u$ is at most $2^{-q} \cdot |X|$. By symmetry, we have the same bound on the number of such $y\in Y$ and the lemma follows. 
    
    To obtain this bound, we run \autoref{alg:2} on $T' = (X'\times Y')\cap T_u$, with $X'$ and $Y'$ taking their values right before processing node $u$ in \autoref{alg:1}. From this, we obtain a tree $\calT$ with the inner nodes labeled by triangles that are subsets of $T'$ not covered yet, and with the edges labeled by rectangles in $\calR$. It satisfies the following properties: 
    \begin{enumerate}
        \item First, the set of edge label rectangles cover $T'$, i.e., for every $(x,y)$ in $T'$, there is an edge labeled $R_{u,v}$ s.t. $(x,y) \in R_{u,v}$. \item Second, if $(x,y) \in T'$ is in triangle $T_1$ at node $t_1$, and $(x,y)$ is in triangle $T_2$ at node $t_2$, then either $t_1$ is a descendant of $t_2$ or vice-versa, and $(x,y)$ is in every triangle labeling nodes between $t_1$ and $t_2$. \item Furthermore, for every $x$, there is a unique path from the root to some leaf $t$ in $T'$ such that $\{R_{u,v}\,\mid \, R_{u,v}\textsf{ edge labels of path}\}$ covers $T'^x$. 
    \end{enumerate}

    For any node $t$ of $\calT$ whose path from the root is labeled by $R = (R_1,R_2,\dots,R_\ell)$, the block-depth of $t$ is the number of unique blocks mentioned by the rectangles in $R$. Then, the block-depth of an input $x\in X'$ is the maximal block-depth over nodes that $x$ is consistent with.
    
    Note that for every leaf $t$ in $\calT$, property 3 above implies that the block-depth of $t$ is an upper bound on $bw(T'_u,x)$ for any $x\in X'$ consistent with $t$. To prove the lemma, we thus only have to bound the number of elements $x\in X'$ that are consistent with leaves of block-depth greater than $\threshold$. Before doing so, we \ifthenelse{\boolean{conferenceversion}}
{observe}{prove} a claim that will allow us to simplify~$\calT$.
    \begin{claim}
        Let $t$ be a vertex in $V(\calT)$. If $t$ has out-degree greater than one, then every child of $t$ has block-depth greater than the block-depth of $t$.
    \end{claim}

\ifthenelse{\boolean{conferenceversion}}
{}{
    \begin{proof}
        Let $t'$ be a vertex that has the same block-depth as its parent $t$. The edge from $t$ to $t'$ must be labeled by a rectangle $R_{u,v}$ with vertices from blocks $i=B(u)$ and $j=B(v)$, both already mentioned on the path from the root to $t$. In particular, this means that all $x\in X'$ consistent with~$t$ agree on the $i$th and $j$th coordinate. 
        
        Now, let $y \in Y$ chosen when processing $t$ in \autoref{alg:2}. 
        If there exists an $x'\in X'$ consistent with $t$ such that $(x',y) \in R_{u,v}$, i.e., such that $(x'_i,y_i) = u$ and $(x'_j,y_j) = v$, this implies that for every $x\in X'$ consistent with~$t$ it holds that $(x_i,y_i) = u$ and $(x_j,y_j) = v$ since $x$ agrees with $x'$ in the  $i$th and $j$th coordinate, and therefore $(x,y) \in R_{u,v}$. Hence $T^y$ in \autoref{alg:2} is fully covered by $R_{u,v}$ and the out-degree of $t$ is 1.
    \end{proof}
}

    We can now modify $\calT$ by removing every vertex in $\calT$ with a child of the same block-depth and replacing it with that child. This leaves us with a new tree $\calT '$ where every path from the root to a leaf has strictly increasing block-depth. Note that the block-depth of an input $x\in X'$ is the same in $\mathcal{T}'$ as in $\mathcal{T}$. 
    In order to conclude the proof, it is now sufficient to \ifthenelse{\boolean{conferenceversion}}
{note the following straightforward claim}{prove the following claim}.
    
\begin{claim}
    \label{claim:size-bound}
    The number of vertices at block-depth $d$ in $\calT '$ is at most $\left(\frac{\sqrt{n}}{2}\right)^d$. 
\end{claim}

\ifthenelse{\boolean{conferenceversion}}
{}{
\begin{proof}
    We prove this by induction over the block-depth $d$. For $d \leq 1$, the claim follows since the only vertex at block-depth $0$ is the root of $\calT '$, and there are no vertices at block-depth $1$.

    For the induction step, 
    note that the children of a vertex $t \in V(\calT')$ at block-depth $d'$ either have block-depth $d' + 1$ or $d' +2$. We will argue that any vertex $t \in V(\calT')$ at block-depth $d'$ can have at most $\sqrt{n}/4$ children of 
    block-depth $d' + 1$ and at most $n/4$ children of block-depth $d' + 2$. 
    This would imply the claim since
    we can conclude that, for $d \geq 2$, there are at most
    \begin{align*}
       \frac{\sqrt{n}}{4}\left(\frac{\sqrt{n}}{2}\right)^{d-1}
        +
       \frac{n}{4}\left(\frac{\sqrt{n}}{2}\right)^{d-2}
       \leq\left(\frac{\sqrt{n}}{2}\right)^d 
    \end{align*}
    vertices at block-depth $d$, where we apply the induction hypothesis to bound the number of vertices at level $d -1$ and $d -2$.

    Let $t \in V(\calT')$ be a vertex at block-depth $d'$. Consider the $y$ chosen when processing this vertex in \autoref{alg:2} and the rectangles $R_{u_1,v_1},\dots,R_{u_d,v_d}$ chosen to be a minimal covering for $T^y$. These rectangles are the labels of the out-edges of $t$. Since the block-width of $y$ is at most $2\threshold$, there is a set of indices $S\subseteq [k]$, with $\lvert S\rvert \leq 2\threshold$, such that $B(\{u_1,v_1,\ldots,u_{d},v_{d}\})\subseteq S$.
    
    We first argue that $t$ has at most $\sqrt{n}/4$ children of 
    block-depth $d' + 1$. 
    A child of $t$ which is adjacent to an edge labeled $R_{u,v}$ is at level $d' + 1$ if and only if exactly one of $B(u)$ or $B(v)$ is a block that already appeared in the path from the root of $\calT$ to $t$. 
    For each pair of indices $i,j \in S$, where $i$ already appears in the path from the root to $t$ and $j$ does not, we bound the number of rectangles $R_{u,v}$ such that $\{i,j\} = B(\{u,v\})$.
    Since $i$ appears earlier in the path, all $x \in X'$ consistent with $t$ agree on the $i$th coordinate and thus $x_i$ is fixed for all rectangles $R_{u,v}$ labeling the out-edges of $t$ with $ B(\{u,v\}) = \{i,j\}$. Since $G$ is $s$-\ndense, and both $y_i$ and $y_j$ are clearly also fixed in $T^y$, then there are at most 
    $s$ 
    such rectangles. 
    Considering every possible pair of indices $(i,j)$, of which there are at most $\binom{|S| }{ 2} \leq 2\threshold^2$, we conclude that $t$ has at most $2\threshold^2 s\leq \frac{\sqrt{n}}{4}$ children of 
    block-depth $d' + 1$.

    It remains to argue that $t$ has at most ${n}/4$ children of 
    block-depth $d' + 2$. A child of $t$ which is adjacent to an edge labeled $R_{u,v}$ is at level $d' + 2$ if and only if neither $B(u)$ nor $B(v)$ is a block that appears in the path from the root of $\calT$ to $t$. For each pair of indices $i,j \in S$, where neither $i$ nor $j$ appears in the path from the root to $t$, and for each $x_i \in [\sqrt{n}]$ we use that $G$ is $s$-\ndense (again using that $y_i$ and $y_j$ are fixed in $T^y$) to obtain that the number of rectangles $R_{u,v}$ labeling out-edges of $t$ such that  $B(u) = i$, $B(v) = j$ and $(x_i,y_i) = u$ is at most 
    $s$.
    Considering every possible pair of indices $(i,j)$ and every $x_i \in [\sqrt{n}]$, we conclude that $t$ has at most $2 q^2 \sqrt{n} s \leq n/4$ children of 
    block-depth $d' + 2$. 
\end{proof}
}

Given this claim we can conclude that \autoref{lemma:few x are assigned per node} holds.
Indeed, every $x\in X'$ mapped by $\mu$ must be consistent with a leaf in $\calT$ of block-depth larger that $\threshold$, and
since block-depth can only increase by $1$ or $2$ from parent to child in $\calT$, every $x\in X'$ consistent with a leaf of block-depth greater that $\threshold$ is consistent with a vertex of block-depth exactly $\threshold+1$ or $\threshold +2$. 
Furthermore, only $\lvert X \rvert/\sqrt{n}^\ell$ inputs $x\in X'$ can be consistent with a vertex $t\in V(\calT')$ at block-depth $\ell$ since all $x\in X'$ consistent with $t$ coincide on the $\ell$ blocks mentioned on the path to $t$. 
From \autoref{claim:size-bound}, it follows that at most
\begin{equation*}
    \left(\frac{\sqrt{n}}{2}\right)^{\threshold +1}
    \cdot
    \frac{\lvert X\rvert}{\left(\sqrt{n}\right)^{\threshold + 1}}
    +
    \left(\frac{\sqrt{n}}{2}\right)^{\threshold+2}
    \cdot
    \frac{\lvert X\rvert}{\left(\sqrt{n}\right)^{\threshold+2}}   
    \leq \cdot2^{-\threshold}\lvert X \rvert 
\end{equation*} elements of $\lvert X\rvert $ are mapped to $u$ by $\mu$. By symmetry the same holds for $\lvert Y\rvert$ and the lemma follows.

\section{Lower Bound for Bounded Depth Resolution over Parities}
\label{sec:binCliqueResLin}

In this section, we prove our second result, stated formally below.

\begin{theorem}
    \label{thm:reslin-bound-random-graph}
    Let $\gamma>0$ be a constant, let $p\in [1-n^{-\gamma},1]$ and $k \in\mathbb{N}$  be functions of $n$, where $k < \exp(n^{o(1)})$ and let $\Gbf\sim \calG(n,p,k)$. Then asymptotically almost surely, for any $\ResLin$-refutation of depth $D$ and size $S$ that refutes $\binCliqueGk$ it must hold that
    
    \begin{equation*}
    D\sqrt{\log S}\geq
        \begin{cases}
            \Omega(n^{3\gamma/2}), \text{ if } \gamma < 1/2\\
            \Omega(\frac{n^{\frac{1+\gamma}{2}}}{\sqrt{\log kn}}), \text{ if } 1/2\leq \gamma < 1\\
            \Omega(n^{1-\varepsilon}), \text{ if } \gamma \ge 1 
        \end{cases}
    \end{equation*}
    for any constant $\varepsilon > 0$.
\end{theorem}
We start by verifying that \autoref{thm:ResLin-informal} from the introduction follows from this theorem. Let $\gamma= 1/3$ and $k=n^{1/3+\varepsilon/100}$ for $\varepsilon>0$ given in \autoref{thm:ResLin-informal}.
    Then $N=kn=n^{4/3+\varepsilon/100}$, and $\binCliqueGk$ has $m= k\log N = 
    O(n^{1/3+\varepsilon/100}\log n)$. 
    So, if 
    \begin{equation*}
    D\leq m^{1.5-\varepsilon}=O(n^{(1/3+\varepsilon/100)(1.5-\varepsilon)}\log^{1.5-\varepsilon} n)=o(n^{3\gamma/2-\varepsilon/4})\,,
    \end{equation*}
    we have that $\log S= \Omega(\textrm{poly}(n))=\Omega(\textrm{poly}(N))$. Therefore, sampling $\bf G$ using $p=n^{-1/3}$ and $k=n^{1/3+\varepsilon/100}$  \aas yields an unsatisfiable formula requiring size $S= 2^{N^{\Omega(1)}}$ to refute in depth at most $m^{1.5-\varepsilon}$, as required.

To prove \autoref{thm:reslin-bound-random-graph} we prove the slightly more general statement below. 
 
\begin{theorem}
    \label{thm:reslin-bound-bcn}
    Let $\edgeProbResLin>0$, $0<\beta<\betabound$ be real numbers and let $\CNbound\leq \rbound$ be a positive integer. Given a graph $G$ with $(\edgeProbResLin,\beta,\CNbound)$-bounded common neighborhoods and a $\ResLin$-refutation of $\binCliqueGk$ of depth $D$ and size $S$, it must hold that 
    \begin{equation*}
        D\sqrt{\log S}\geq\Omega(\CNbound\cdot\min(\sqrt\CNbound,1/\beta))
    \end{equation*}
\end{theorem}

The remainder of this section is dedicated to the proof of this theorem. But first we show that \autoref{thm:reslin-bound-random-graph} follows from \autoref{thm:reslin-bound-bcn}.

\begin{proof}[Proof of \autoref{thm:reslin-bound-random-graph}]
    Let $\edgeProbResLin = p \geq 1-n^{-\gamma}$. We choose $\delta$ differently according to $\gamma$. Let $\delta = \gamma$ if $\gamma <1$, otherwise, let $\delta = 1 - \varepsilon$. Now let $\CNbound = n^\delta/100$ and let $\beta = \sqrt\frac{18\log(kn)}{n^{1-\delta}}$.
    By \autoref{lemma:neigbor-concentrates}, $\random{G}$ will asymptotically almost surely have $(\edgeProbResLin,\beta,\CNbound)$-bounded common neighbourhoods. Note that $R \leq \rbound $
    and so we can apply \autoref{thm:reslin-bound-bcn} and obtain the lower bound 
    \begin{equation*}
        D\sqrt{\log S}\geq\Omega(\CNbound\cdot\min(\sqrt{\CNbound},1/\beta))  \,.
    \end{equation*}
    
    We split the analysis into cases depending on $\gamma$. If $\gamma < 1/2$, 
    note that $1/\beta \geq n^{-\delta/2} = n^{-\gamma/2}  \geq \sqrt{\CNbound}$ for large enough $n$ and $k< \exp(n^{o(1)})$,
    and therefore we get the first lower bound, that is,
$ D\sqrt{\log S}\geq \Omega(n^\gamma\cdot n^{\gamma/2}) = \Omega( n^{3\gamma/2})$.

If $\gamma \geq 1/2$, note that 
$1/\beta \leq n^{(1-\delta)/2}\leq n^{\delta/2} = 10\sqrt{\CNbound}$ for large enough $n$ and using the assumption $k< \exp(n^{o(1)})$.
This implies that we get the bound
\begin{equation*}
        D\sqrt{\log S}\geq \Omega\left(n^\gamma\cdot \frac{n^{\frac{1-\gamma}{2}}}{\sqrt{\log nk}} \right) 
        = \Omega\left(\frac{n^{\frac{1+\gamma}{2}}}{\sqrt{\log nk}} \right)
    \end{equation*}
for $1/2\leq \gamma < 1$ and the bound
   \begin{equation*}
        D\sqrt{\log S}\geq\Omega\left(n^{\delta}\cdot \frac{n^{\frac{1-\delta}{2}}}{\sqrt{\log nk}} \right) 
        = \Omega\left(\frac{n^{\frac{1+\delta}{2}}}{\sqrt{\log nk}} \right)
        = \Omega\left(\frac{n^{1 - \varepsilon/2}}{\sqrt{\log nk}} \right)
        \geq \Omega\left(n^{1 - \varepsilon}\right)
    \end{equation*}
for $\gamma > 1$, where for the last inequality we assume $n$ large enough and $k< \exp(n^{o(1)})$.

\end{proof}

The proof of \autoref{thm:reslin-bound-bcn} adapts the proof by Byramji and Impagliazzo \cite{BI25bphp} of a $\ResLin$ size depth tradeoff for the BPHP. They do this by showing that in a small $\ResLin$ refutation of BPHP, it is possible to repeatedly perform long random walks without reaching a sink. 

\subsection{Technical Lemmas and Definitions}

Here we collect some technical preliminaries that are used only in this section. We start by recalling the notion of safe systems, introduced by \cite{EGI24regular}. 

We say that a set of linear forms $F$ over variables $x_{i,j} \in ({\{0,1\}}^{\log n})^{[k]}$ is \emph{dangerous} if the number of distinct blocks from which variables appear in $F$ is smaller than the size of $F$. If no subset of $F$ is dangerous we say that $F$ is \emph{safe}. It turns out that safe sets have a very nice characterization making them nice to work with.

\begin{lemma}[\cite{EGI24regular}]
\label{lemma:safeSystem-charactarisation}
Let $L$ be a collection of $k$ independent linear forms and $M$ the corresponding coefficient matrix. $V$ is safe if and only if we can pick $k$ variables, no two from the same block, such that the corresponding columns in $M$ are linearly independent.
\end{lemma} 

To handle non-safe systems we use the notion of closure. The \emph{closure} of $F$, $\cl(F) \subseteq [k]$ is defined as the minimal set (with respect to inclusion) of blocks such that $F[\setminus \cl(F)]$ is safe, where $F[\setminus S]$ denotes the set of linear forms $F$ after setting variables belonging to a block in $S$ to zero. It turns out that the closure is uniquely defined \cite[Lemma 4.1]{EGI24regular}. In an abuse of notation we will, for a linear system $\Psi$, write $\cl(\Psi)$ to refer to the closure of 

\begin{lemma}[\cite{EGI24regular}]
\label{lemma:safe-dimension-bound} If $F$ is a collection of linear forms, then $\setsize{\cl(F)} + \mathrm{dim}(\langle F[\setminus\cl(F)]\rangle) \leq
\mathrm{dim}(F)$.
\end{lemma}

To assign variables in a way that satisfies linear systems, we need the concept of affine restrictions. A (block respecting) \emph{affine restriction} $\rho$ is a partial assignment that for some set $F\subseteq [k]$ sets all variables $x_{i,j}$, $i\in [k]\setminus F$ as an affine function of the unassigned variables $x_{i,j}$, $i\in F$. In this paper all affine restrictions will be block respecting and will for brevity just be called affine restrictions. For a linear system $\Psi$ we use $\Psi|_\rho$ to denote the system after substituting variables according to $\rho$.

We introduce the search problem $\mathrm{NonEdge}^G_M \subseteq (\{0,1\}^{\log n})^{[k]\setminus B(M)}\times \binom{V(G)}{2}$, where $M$ is some set of vertices of $G$ such that vertices in $M$ from distinct blocks in $G$ always have edges between them. As input we are given an assignment to the blocks outside of the support of $M$ such that every selected vertex lies in $N^\cap(M,[k]\setminus B(M))$. We are then asked to output two selected vertices not connected by an edge. When the graph $G$ is clear from context, we simply write $\mathrm{NonEdge}_M$. Observe that $\mathrm{NonEdge}_M^G$ is equivalent to $\search_{\binClique{G'}{k - |B(M)|}}$ where $G'$ is the subgraph of $G$ induced by $N^\cap(M,[k]\setminus B(M))$.  
Thus, any affine-DAG solving the falsified clause search problem also solves $\mathrm{NonEdge^G_\emptyset}$, as such it is enough to find lower bound for this problem. 
The following lemma was stated in $\cite{BI25bphp}$ with $A_1=A_2=\dots=A_n$. Their proof still works for or slightly more general formulation.

\begin{lemma}[{\cite[Lemma 4.6]{BI25bphp}}]
    \label{lemma:rank-probability-relation}
    Let $\Psi$ be a linear system on $(\{0,1\}^{\log n})^k$ whose rank is $r$. For each $i \in [k]$ let $A_i \subseteq \{0,1\}^{\log n}$ be such that $|A_i|\geq 2n/3$. Let $\bm x_i \sim A_i$ independently for each $i \in [k]$ Then \begin{equation*}
        \Pr[\bm x\textrm{ satisfies }\Psi] \leq \left(\frac{3}{4}\right)^r.
    \end{equation*} 
\end{lemma}

Lastly, we will later need to use a Chernoff bound on the sum of dependent variables.  For this the following lemma will be useful.

\begin{lemma}[{\cite[Lemma 17.3]{MU17}}]
\label{lemma:conditional-chernoff}
Let $\bm X_1,\bm X_2\dots,\bm X_n$ be a sequence of random variables in an arbitrary domain and let $\bm Y_1,\bm Y_2,\dots,\bm Y_n$ be a sequence of binary variables with the property that $\bm Y_i = \bm Y_i(\bm X_1,\dots,\bm X_i)$. If
\begin{equation*}
    \Pr[\bm Y_i =1\mid \bm X_1,\dots, \bm X_{i-1}] \leq p,
\end{equation*}
then 
\begin{equation*}
    \Pr\left[\sum_{i=1}^n \bm Y_i>k\right]\leq \Pr \left[\mathrm{Ber}^n(p)>k\right].
\end{equation*}
\end{lemma}

\subsection{Proof of \autoref{thm:reslin-bound-bcn}}

Throughout this section we let $\edgeProbResLin$, $\beta<\betabound$ be positive real values and $\CNbound<\rbound$ a positive integer. Furthermore, we assume $G$ is a fixed graph with $(\edgeProbResLin,\beta,\CNbound)$-bounded common neighborhoods. Note that this choice of parameters ensures that for every vertex set $|M|\leq R$, $|N^\cap(M,i)|\leq2n/3$.

The proof of \autoref{thm:reslin-bound-bcn} follows the same strategy as \cite{BI25bphp}. Given an affine-DAG solving $\mathrm{NonEdge}_M$ we wish to find an affine-DAG of much smaller depth solving a slightly smaller instance $\mathrm{NonEdge}_{M^*}$; repeating this process shows that the original depth must have been large. Formally, we use the following lemma which we will spend the remainder of the section proving.

\begin{restatable}{restatablelemma}{lemmareduction}
    \label{lemma:reduction}
    Let $M\subseteq V(G)$ be a set of vertices 
    and suppose there is an affine-DAG $\affineDAG$ solving $\textrm{NonEdge}_M$ whose depth is at most $D$ and size is at most $S$. Then 

    $D>(R-|M|)/8$.
    
    Moreover, for some universal constant $C$ it holds that if $|M|+8d\leq R$ for $d=\lfloor \sqrt{\log S} \cdot \min(\sqrt\CNbound,1/\beta)\rfloor$ 
    then there exists a set
    $M^*\supseteq M$ of size $|M^*|\leq |M|+C\log S$ and an affine-DAG $\affineDAG'$ solving $\textrm{NonEdge}_{M^*}$ whose depth is at most $D-d$ and size is at most $S$. 
\end{restatable}

First we show how \autoref{thm:reslin-bound-bcn} follows from \autoref{lemma:reduction}.

\begin{proof}[Proof of \autoref{thm:reslin-bound-bcn}]
A $\ResLin$-refutation of $\binCliqueGk$ of depth $D$ and size $S$ immediately gives an affine $\affineDAG$ solving $\mathrm{NonEdge}_\emptyset$. Let $d=\lfloor \sqrt{\log S} \cdot \min((1-\edgeProbResLin)^{-1/2},1/\beta)\rfloor$ we start by considering the case when $8d\geq R/2$. Then \begin{equation*}
    R\leq 16d \leq \sqrt{\log S} \min(\sqrt \CNbound,1/\beta)\leq \sqrt{\log S}\sqrt{\CNbound}. 
\end{equation*}
This means $\sqrt{\log S} \geq \sqrt{R}$ and the desired bound follows immediately since by
\autoref{lemma:reduction} we have that $D > R/8$.

Now consider the case $8d<R/2$. We can here use \autoref{lemma:reduction} to find a set $M_1$ of size $|M_1|=C\log S$ and an affine-DAG $\affineDAG_1$ of depth at most $D-d$ and size at most $S$ solving $\mathrm{NonEdge_{M_1}}$. 
We keep applying \autoref{lemma:reduction} iteratively in this way, using $M=M_{t-1}$ to find $M_{t} = M^*$ of size $|M_t|\leq tC\log S$ and $\affineDAG_t$ of depth at most $D-dt$ and size at most $S$. We repeat this process as long as the condition of \autoref{lemma:reduction} holds, i.e., as long as $|M_t|+ 8d \leq \CNbound$. 
Thus the final $t$ must satisfy $tC\log S +8d \geq \CNbound$ which means $t \geq \CNbound/2C\log S$. But at that point we have an affine-DAG $\affineDAG_t$ of depth $D-d\CNbound/2C\log S$ solving $M_t$.
Since an affine-DAG cannot have negative depth we find that $D \geq \Omega(d\CNbound/\log S) = \Omega(\CNbound \cdot \min(\sqrt{\CNbound},1/\beta)/\sqrt{\log S})$ and thus the theorem follows.
\end{proof}

Now all that remains is to prove \autoref{lemma:reduction}. We do so by constructing a random walk according to \autoref{alg:PDTsimulator}, over our affine-DAG. We require from this random walk to satisfy several properties. During the walk, we build an affine restriction, which sets relatively few variables and satisfies all linear systems labeling nodes on the walked path. 
We will ensure that the walk visits all nodes with the same probability, as answering according to a uniformly sampled fixed input.
We show that, with a significant probability, the walk does not witness a missing edge. On the other hand, if the affine-DAG is not sufficiently large, some node $v$ must be visited with a significant probability in one of these successful runs. This is only possible if the rank of the system $\Psi$ labeling $v$ is not too large. 

We can then use the fact that a small rank $\Psi$ can be satisfied
by restricting a few blocks. We combine this with the fact that $v$ was reached in a successful run to find a valid instance of $\mathrm{NonEdge}_{M^*}$ that is solved by the dag rooted at $v$, i.e., a DAG of lower depth as desired.

We start by presenting \autoref{alg:PDTsimulator}. This is essentially the same algorithm as \cite[Algorithm 1]{BI25bphp} with two main differences. First, we require $x_i$ to lie in $N^\cap(M,i)$ rather than an arbitrary set $A$. Second, we say that the algorithm \emph{fails} if it assigns blocks in a way that could select vertices not connected by an edge rather than if it finds a collision between pigeons.  

\begin{algorithm2e}[t!]
\caption{Random walk simulator}
\label{alg:PDTsimulator}
    \KwIn{$M$: Set of vertices. 
    $T$: Parity decision tree.
    }
    $C\leftarrow []$\tcp*{List of vertices used during simulation}
    $F\leftarrow [k]\setminus B(M)$\tcp*{Unassigned blocks} 
    $L\leftarrow \emptyset$\tcp*{Collection of equations}
    $v\leftarrow \text{root of T}$\;
    \While{$v$ is not a leaf}{ \label{alg-line:while-loop}
    $P'\leftarrow $ query at $v$\;
    $P \leftarrow P'$ after substituting according to $L$\;
    
    \uIf{$P$ is a constant, $b\in \{0,1\}$\label{alg-line:P-is-implied}}{
        update $v$ according to $b$\;
    }\Else{
    $(i,j)\leftarrow \min\{(i,j)\in F\times[{\log n}]\mid x_{i,j}\text{ appears in } P\}$\label{alg-line:set(i,j)}\;
    pick $y$ uniformly from $N^\cap(M,i)$\;
    $L\leftarrow L\cup \{x_{i,h} = y_h\mid h\in[{\log n}]\setminus \{j\}\}$\label{alg-line:fixingx_i,h_j}\;
    \uIf{$y^{\oplus j}\notin N^\cap(M,i)$}{\label{alg-line:bit-flip-check}
        $L\leftarrow L\cup\{x_{i,j} = y_j\}$\;
        append $\{(i,y)\}$ to $C$\;
    }\Else{
    pick $b\in \{0,1\}$ uniformly at random\;
    $L\leftarrow L\cup\{P=b\}$\;
    append $\{(i,y),(i,y^{\oplus j})\}$ to $C$\;
    update $v$ according to $b$\;
    }
    $F\leftarrow F\setminus \{i\}$
    }
  }
  \uIf{there exists $j_1\neq j_2$ such that $ (i_1,y_1)\in C_{j_1}, (i_2,y_2)\in C_{j_2}$ and there is no edge between vertex $y_1$ in block $i_1$ and $y_2$ in block $i_2$}{
  \Return FAIL\;
  }\Else{
  \Return $v,C,L,F$\;
  }
\end{algorithm2e}

Let us state two properties of \autoref{alg:PDTsimulator} that follow immediately from the analysis of \cite{BI25bphp}, the difference lying in the fact that we use $N^\cap(M,i)$ as a set of valid assignments to the $x_i$ rather than an arbitrary set $A$.

\begin{lemma}[{\cite[Lemma 4.1]{BI25bphp}}]
    \label{lemma:simulation-properties}
    At the end of every iteration of the while loop in \autoref{alg:PDTsimulator}
    \begin{enumerate}
        \item the equations in $L$ uniquely determine $x_{i,j}$ for all $i\in[k]\setminus F$, $j\in [{\log n}]$ as an affine function of $x_{i,j}$, $i\in F$, $j\in [{\log n}]$;
        \item fixing the variables in $x_{i,j}$, $i\in F$, $j\in [{\log n}]$ and setting the remaining variables as determined by $L$ ensures that $x_{i'}\in N^\cap(M,i')$ for all ${i'}\in [k]\setminus F$;
        \item $L$ implies all parity constraints on the path from the root to the current node.
    \end{enumerate}
\end{lemma}

\begin{lemma}[{\cite[Lemma 4.2]{BI25bphp}}]
\label{lemma:PDTwalkRelation}
Let $W_v$ be the event that node $v$ of $T$ is visited by \autoref{alg:PDTsimulator}, Let $\bm x_i \sim N^\cap(M,i)$ independently for each $i \in [k] \setminus B(M)$, and $V_v$ be the event that running $T$ on $\bm x$ reaches $v$. Then for every $v \in T$, we have $\Pr[W_v] = \Pr[V_v]$. 
\end{lemma}

Now, we will show that the algorithm does not fail with a non-negligible probability (i.e., the walk does not witness a missing edge).

\begin{lemma}
    \label{lemma:pdtSuccessProb}
    Let $T$ be a depth-$d$ parity decision tree and $M$ a set of vertices. If $|M|+8d\leq \CNbound$, then the probability that \autoref{alg:PDTsimulator} does not return FAIL when run on $T$ and  $M$ is at least $\exp(-\Omega(d\beta+d^2(1-\edgeProbResLin))$.
\end{lemma}

\ifthenelse{\boolean{conferenceversion}}
{}{
\begin{proof}
    To bound the probability of failure we consider two cases depending on the number of iterations of the while loop on \autoref{alg-line:while-loop}. First consider the case where the algorithm runs for more than $4d$ iterations. We wish to show that the probability of this happening is small. To see this,  notice that in every iteration we update $v$ except if we reach \autoref{alg-line:bit-flip-check} with $y^{\oplus j}\notin N^\cap(M,i)$. But since $|M|\leq \CNbound\leq \rbound$, it follows by the bounded common neighborhood property of the underlying graph  that $|N^\cap(M,i)| > 2n/3$. Since $y$ is picked uniformly from $N^\cap(M,i)$, the probability of $y^{\oplus j}\notin N^\cap(M,i)$ is bounded above by $1/2$. Let $\bm X_t$ be the indicator of the event that we update $v$ on iteration $t$. Then, since $\Exp [\bm X_t \mid \bm X_{t-1},\dots, \bm X_1] \geq 1/2$ we can apply \autoref{lemma:conditional-chernoff} with $\bm Y_i = \bm X_i$ to find $\Pr[\sum_{t=1}^{4d}\bm X_t<d] \leq \Pr[\mathrm{Ber}^{4d}(1/2)<d] \leq \exp(-d/4)$. 
    
    We now wish to bound the probability of finding a non-edge in the first $4d$ iterations. We start by bounding the probability in the $t$th iteration. Let $i$ and $j$ be the indices selected on \autoref{alg-line:set(i,j)}. In every iteration, we append a set of at most size $2$ to $C$, and so by the $t$th iteration the union of all sets in $C$ contains at most $2(t-1)$ elements. Call this union $C_t$. By the bounded common neighborhood property, we know that the number of vertices in the $i$th block that lie in the common neighborhood of vertices in $C_t$ and $M$ is at least $(1-\beta)\edgeProbResLin^{|M|+2(t-1)}$ while the size of $N^\cap(M,i)$ is at most $(1+\beta)\edgeProbResLin^{|M|}$. 
    Thus, the probability of $y$ having an edge to a vertex in $C_t$ is at least $\frac{1-\beta}{1+\beta}\edgeProbResLin^{2(t-1)}\geq(1-2\beta)\edgeProbResLin^{2(t-1)}$. Similarly, the probability of $y^{\oplus j}$ either not lying in $N^\cap(M,i)$ or having an edge to some vertex in $C_t$ is also bounded below by this. By a union bound, the probability of not failing at the $t$th iteration is at least $2(1-2\beta)\edgeProbResLin^{2(t-1)}-1$. Note that, for $t\leq4d$, 
    \begin{equation*}
    \edgeProbResLin^{2(t-1)}\geq \edgeProbResLin^{8d}\geq \edgeProbResLin ^R\geq \edgeProbResLin^{1/2(1-\edgeProbResLin)} \geq \edgeProbResLin^{-1/2\log \edgeProbResLin} = 1/\sqrt 2. \,
    \end{equation*} 
    Hence, using the fact that $\beta <1/20$, we have that
    \begin{equation*}
        2(1-2\beta)\edgeProbResLin^{2(t-1)}-1 \geq \sqrt2(1-1/10)-1\geq 1/4 \,.
    \end{equation*}
 
    The probability of not failing after $4d$ iterations is thus at least
    \begin{align*}
        \prod_{t=1}^{4d}\left(2(1-2\beta)\edgeProbResLin^{2(t-1)}-1\right)
        &\geq \prod_{t=1}^{4d}\exp(4((1-2\beta)\edgeProbResLin^{2t}-1))\\ 
        &\geq\exp(4\sum_{t=1}^{4d}((1-2\beta)\edgeProbResLin^{2t}-1))\\
        &\geq\exp(4\sum_{t=1}^{4d}((1-2\beta)(1-2t(1-\edgeProbResLin))-1))\\ 
        &\geq \exp(-8\sum_{t=1}^{4d}(\beta+t(1-\edgeProbResLin)))\\
        &\geq \exp(-32d\beta-64d^2(1-\edgeProbResLin)) \,,
    \end{align*}
    where for the first inequality, we use that $1-x \geq \exp(-2x)$ for $x<0.75$. In total, the probability of both not failing within the first $4d$ iterations and terminating within $4d$ steps is bounded from below by $\exp(-32d\beta-64d^2(1-\edgeProbResLin))-\exp(-d/4)\geq \Omega(\exp(-32d\beta-64d^2(1-\edgeProbResLin)))$, using the fact that $\beta \leq \betabound$ and $8d\leq\CNbound\leq\rbound$.
    \end{proof}
}

We now wish to show that it is possible to use a successful run of \autoref{alg:PDTsimulator} to find an affine restriction satisfying the reached node. 

\begin{lemma}
    \label{lemma:assignment-extraction}
    Let $v,C,L,F$ be returned by a successful run of \autoref{alg:PDTsimulator} when run on $T$ and $M$. Let $\Psi$ be the linear system labeling the node in the affine-DAG corresponding to $v$, and let $r$ be the rank of $\Psi$. If $|M|+2r\leq \CNbound $, then there exists a set $F'\subseteq [k]$ and an affine restriction $\rho$ fixing blocks $[k]\setminus F'$ satisfying the following conditions:
    \begin{enumerate}
        \item The number of blocks fixed by $\rho$, say $s$, is at most $r$.
        \item $\Psi$ is implied by $\Psi_\rho$ (the linear system equivalent to the affine restriction $\rho$).
        \item There exists a set $M'\subseteq V(G)\setminus M$ such that  $|M'|\leq 2s$ and for any assignment to $F'$, if we set the blocks $[k]\setminus F'$ according to $\rho$, all the vertices assigned to $[k]\setminus F'$ lie in $M'$ and form a complete subgraph of $G$.
    \end{enumerate}
\end{lemma}

\ifthenelse{\boolean{conferenceversion}}
{}{
\begin{proof}
    We start by fixing the variables in $\cl(\Psi)\cap F$. Let $\hat M$ be the set of vertices assigned by the simulation to the blocks in $\cl(\Psi)\setminus F$, i.e., the $y$th vertex of the $i$th block for every $(i,y)\in C_t$ for some $t$ and $i\in \cl(\Psi)\setminus F$. Note that $|\hat M|\leq 2|\cl(\Psi)\setminus F|$. We now assign $x_i$, for every $i\in \cl(\Psi)$, such that $x_i\in N^\cap(M\cup\hat M,i)$ and all vertices selected form a complete subgraph. Since $G$ has $(\alpha, \beta, R)$-bounded common neighborhoods, this can be done as long as $|M\cup\hat M|+|\cl(\Psi)\cap F|\leq \CNbound$. But this holds since 
    \begin{align*}
        |M\cup\hat M|+|\cl(\Psi)\cap F|\leq |M|+ 2|\cl(\Psi)\setminus F|+|\cl(\Psi)\cap F|\leq |M| + 2|\cl(\Psi)|\leq |M|+2r \leq \CNbound
    \end{align*}

    Set all blocks in $F\setminus \cl(\Psi)$ arbitrarily and extend this to a full assignment $x$ according to $L$. Since all blocks in $F$ have been set, this extension exists and is unique by \autoref{lemma:simulation-properties}. Furthermore, $L$ implies all parity constraints on some path to node $v$ which in particular means that $L$ implies $\Psi$. Thus, $x$ satisfies $\Psi$. Let $\rho_1$ be the partial assignment setting every block in $\cl(\Psi)$ according to $x$. By construction this assignment will not falsify any equation in $\Psi$. 
    For the second step, we consider $\Psi' = \Psi|_{\rho_1}$. Since $\cl(\Psi)$ is fixed by $\rho_1$, $\Psi'$ is safe. By \autoref{lemma:safeSystem-charactarisation}, there exists a set of variables of size $\mathrm{dim}(\Psi)$ belonging to different blocks and whose corresponding columns in $\Psi$ are linearly independent. Let $X$ be this set of variables and $B$ the corresponding set of blocks. We now wish to fix all but one variable of every block in $B$, leaving the variables in $X$ untouched, such that no matter what the variables of $X$ are set to, the vertices in the blocks $B$ all have edges to each other and to all vertices in $M$ and the ones fixed by $\rho_1$. We do this one block at a time. When we get to block $i$, we fix ${\log n}-1$ variables, such that both vertices $\{v_i,v_i'\}$, we could potentially select when setting the last variable, lie in the common neighborhood of $M$, of the vertices selected by $\rho_1$ and of all previous $\{v_{i'},v'_{i'}\}$. By the pigeonhole principle, we can find such an assignment to the ${\log n}-1$ variables of block $i$, as long as the common neighborhood of all potential previous selections at block $i$ is larger than $n/2$. But the number of previously selected vertices is at most $|M|+|\cl(\Psi)|+2(\mathrm{dim}(\Psi')-1)$ which by \autoref{lemma:safe-dimension-bound} can be bounded by 
    \begin{equation*}
        |M|+|\cl(\Psi)|+2(\mathrm{dim}(\Psi')-1) \leq |M|+2(|\cl(\Psi)|+\mathrm{dim}(\Psi')-1)\leq |M|+2r < R.
    \end{equation*}
    This means that the common neighborhood will have size greater than $n/2$ as desired.
    
    Call this partial assignment of $({\log n}-1)|B|$ variables $\rho_2$. To find the final $\rho$ we use the fact that $\Psi'$ is safe to solve $\Psi$ expressing the variables in $X$ as an affine function of the variables outside of $B$. We now let $\rho$ be this affine restricting combined with $\rho_1$ and $\rho_2$. 
\end{proof}
}

With this we have both shown that the probability of succeeding is significant and that when we succeed, we can extract an affine restriction satisfying the system labeling the reached node. We can now combine these facts to \ifthenelse{\boolean{conferenceversion}}{obtain \autoref{lemma:reduction}. For the details of this argument see the full version~\cite{}.}
{prove \autoref{lemma:reduction}, restated here for convenience.}

\ifthenelse{\boolean{conferenceversion}}
{}{
\lemmareduction*
\begin{proof}[Proof of \autoref{lemma:reduction}]

        We start by showing that $D > (\CNbound-|M|)/8$. Assume for contradiction that $D \leq (\CNbound -|M|)/8$ and expand $\affineDAG$ into a parity decision tree $T$ of depth $D$ solving $\mathrm{NonEdge_M}$. Since $|M| + 8D \leq R$, \autoref{lemma:pdtSuccessProb} gives a non-zero probability of \autoref{alg:PDTsimulator} succeeding. But after a successful run returning $(v, C,L,F)$, an input $x$ that reaches $v$ can be constructed by setting the variables in $F$ arbitrarily and setting the remaining variables according to $L$. Since the variables in blocks of $F$ can be set arbitrarily, $v$ cannot be labeled with a valid answer to $\mathrm{NonEdge_M}$ mentioning vertices in $F$. But the vertices assigned outside of $F$ are all connected. This contradicts the assumption that $T$ solves $\mathrm{NonEdge_M}$ and $D>(R-|M|)/8$ must hold. 
        
        To prove the second claim of \autoref{lemma:reduction}, let $T$ be the depth $d$ parity decision tree obtained by starting at the root of the DAG $\affineDAG$, expanding the DAG into a tree, and removing any node at depth greater than $d$. We run \autoref{alg:PDTsimulator} on $T$ and, by \autoref{lemma:pdtSuccessProb}, the algorithm succeeds with probability $q$ at least $\exp(-\Omega(d\beta+d^2(1-\edgeProbResLin))\geq S^{-O(1)}$, here we use the fact that $\CNbound\leq \rbound$. Since the DAG has size at most $S$, one of the nodes $v$ must be reached, after a successful run, with probability at least $q/S \geq S^{-O(1)}$.
        By \autoref{lemma:PDTwalkRelation}, this is also a lower bound on the probability that $\bm x \sim \prod_{i \in [k] \setminus B(M)} N^\cap(M,i)$ ends at this node. By \autoref{lemma:rank-probability-relation}, the rank of $\Psi_v$ is bounded by $O(\log S)$. 

        Now fix some successful run of $T$ that ends at $v$ and apply \autoref{lemma:assignment-extraction} to find a corresponding $F'$, $M'$, $\rho$. Let $M^* = M\cup M'$ and let $\affineDAG'$ be the affine-DAG consisting of the sub-DAG of $\affineDAG$ of vertices reachable from $v$ and restricted by $\rho$. Note that \autoref{lemma:assignment-extraction} ensures that $\rho$ satisfies $\Psi$ and thus the root of $\affineDAG'$ is labeled by the empty system as desired. Furthermore, for any assignment to the blocks $F'$, only selecting vertices within $N^\cap(M^*,[k]\setminus F')$, extending it to all blocks according to $\rho$ can never introduce a non-edge outside of $F'$. Thus $\affineDAG'$ solves $\mathrm{NonEdge_{M^*}}$. Furthermore, note that $v$ can not be a sink in $\affineDAG$ since $\affineDAG'$ solves a non-trivial problem. Thus $v$ must have been reached after exactly $d$ steps in $T$ which means that $\affineDAG'$ cannot have depth higher than $D-d$.
       
\end{proof}
}

\section{Lower Bound for Randomized Communication} \label{sec:binCliqueRandComm}
In this section, we lower bound randomized communication cost required to solve $\binCliqueBlockComm{G}{k}$.
\begin{theorem}
\label{thm:rand-comm-lb}
    For any
    integers $n$ and $k$, and for any real $p \in [0,1]$, 
    if $\bm G$ is a graph sampled from $\calG(n,p,k)$ then \aas randomized communication protocol solving $\binCliqueBlockComm{\bm G}{k}$ has cost 
    \begin{equation*}
        \Omega\left({\min\left(
    \frac{1}{\sqrt{n^{0.05}(1-p)}},
    \,
    \frac{n^{0.225}}{\log(kn)}\right)}\right).
    \end{equation*}
\end{theorem}

   To see that \autoref{thm:Comm-informal} follows we set $p = 1-n^{-1/3}$ and $k=n$. $\bm G$ then has $N=n^2$ vertices and we get the lower bound $\Omega(n^{0.22})=N^{\Omega(1)}$ on the cost of the protocol as required.

For the rest of this section, recall that $S_{G} \coloneqq \binCliqueBlockComm{G}{k} = \{(x, y, uv) \mid x, y \in \Sigma^k,\, uv \not\in E( G),\, \{u,v\} \subseteq V(x,y)\}$, where $\Sigma = \{0,1\}^{\log n / 2}$, $V(x,y) \in V^k$ is the list on nodes described by $x,y\in \Sigma^k$.
\begin{theorem}
\label{thm:CC-lowerbound}
    Let $s,k,n>0$, and $G$ be an $s$-\ndense graph on $k$ blocks of size $n$.
    Then the randomized communication complexity of $S_G$ is $\Omega(\sqrt{n^{0.45}/s})$.
\end{theorem}
\autoref{thm:rand-comm-lb} follows immediately from \autoref{thm:CC-lowerbound} and \autoref{lemma:probability-concentrates}. To see this recall that \autoref{lemma:probability-concentrates} states that \aas a graph $\bm G$ sampled as in \autoref{thm:rand-comm-lb} is $s$-\ndense for $s={\max\{2\sqrt{n}(1-p),9e^2\log(kn)\}}$. From \autoref{thm:CC-lowerbound} it then follows that the randomized communication complexity of $S_G$ is
\[\Omega\left(\sqrt{\frac{n^{0.45}}{s}}\right)=\Omega\left(\min(\frac{1}{\sqrt{n^{0.05}(1-p)}},\frac{n^{0.225}}{\sqrt{\log(kn)}}\right)\]

The rest of this section is dedicated to the proof of \autoref{thm:CC-lowerbound}. By Yao's principle for randomized communication complexity~\cite{yao83}, it suffices to show that every small cost deterministic protocol errs significantly on some input distribution. We will show this for the uniform distribution over inputs $(\bm x, \bm y) \sim \Sigma^k \times \Sigma^k$.

\subsection{Lifting Background}
 We start by introducing some concepts necessary to define subcube-like protocols, which we will use in our analysis.
Every node $v$ in a communication protocol $\protocolPi$ is associated with a rectangle $R_v = X_v \times Y_v$. 
For a random variable $\xbf$ we write its min-entropy as $\Hminof{\xbf}=\min_x \log (1/\Pr[\xbf = x])$. We say that a random variable $\xbf \in \{0,1\}^m$ is \emph{$\gamma$-spread} 
if for every subset $I\subseteq [m]$, the marginal distribution $\xbf_I$ satisfies $\Hminof{\xbf_I}\geq \gamma \setsize{I}$. We define \emph{subcube-like communication protocols} \cite{GoosGJL2025quantum,RiazanovSSY2025searching}.

\begin{definition}
    We say that a rectangle $R=X\times Y \subseteq \{0,1\}^m \times \{0,1\}^m$ is \emph{$\gamma$-subcube-like} with respect to $I,J \subseteq [m]$ if $X_I, Y_J$ are fixed to some values, and the random variables $\bm x_{[m]\setminus I},\bm y_{[m]\setminus J}$ are $\gamma$-spread for $\bm x \sim X$ and $\bm y \sim Y$. We write $\fix (X)=I$ and $\fix (Y)=J$.

    A communication protocol $\protocolPi$ is called \emph{$\gamma$-subcube-like} if all the rectangles associated to the nodes in $\protocolPi$ are $\gamma$-subcube-like.
\end{definition}

In our proof of \cref{thm:CC-lowerbound}, we use the following theorem from~\cite{GoosGJL2025quantum}, which essentially says that, without loss of generality, we can assume the protocol to be subcube-like. 

\begin{theorem}[\cite{GoosGJL2025quantum}, Theorem 29]\label{thm:ggjl-cleanup}
Let $\protocolPi$ be a deterministic communication protocol with input from $X \times Y$. Then, for any constant $\gamma<1$, there is a protocol $\protocolPiAlt$ that  $\gamma$-subcube-like, and has cost $\commCost{\protocolPiAlt}=O(\commCost{\protocolPi}/\epsilon)$ such that $\Pr_{\bm x, \bm y \sim X \times Y}[\protocolPiAlt(\bm x, \bm y) \neq \protocolPi(\bm x, \bm y)] \le \epsilon$. 
\end{theorem}

\subsection{Proof of \autoref{thm:CC-lowerbound}}
\newcommand{\dng}{\textsf{blocks}}
Assume that $\protocolPi$ solves $S_{G}$ with an error $\epsilon$. By \autoref{thm:ggjl-cleanup}, there is a $0.9$-subcube-like protocol $\protocolPiAlt$ of cost $O(|\protocolPi| / \epsilon) \eqqcolon d$ solving 
$S_{G}$ to an error $2\varepsilon$. 
Sample uniform random inputs $\xbf, \ybf \in \Sigma^k = \{0,1\}^{k \log |\Sigma|}$ 
for Alice and Bob. Let the output of $\protocolPiAlt(\xbf,\ybf)$ be an edge connecting nodes from blocks $\bm a \in [k]$ and $\bm b \in [k]$. We denote this edge with $(\bm a \bm x_{\bm a} \bm y_{\bm a}, \bm b \bm x_{\bm b} \bm y_{\bm b})$ meaning that both nodes are described by first specifying the block, and then specifying the number of the node inside the block identified with $\Sigma \times \Sigma$.
Let $\bm \ell$ be the leaf of $\protocolPiAlt$ such that $(\bm x, \bm y) \in R_{\bm \ell}$.
 For a set $I \subseteq [k \log |\Sigma|]$ let the set $\dng(I)$ be the set of all blocks $i \in [k]$ whose bits are mentioned in $I$.

The proof can be summarized as follows. If one of the blocks $\bm a$ and $\bm b$ is never mentioned in the fixed part of the leaf rectangle $R_{\bm \ell}$, then, by density of the graph, the answer is likely to be wrong. Otherwise, at least one bit describing the nodes in $\bm a$ and $\bm b$ is fixed at the leaf. In that case, we argue that since there are at most $d$ blocks that can be mentioned, it is unlikely that a non-edge was discovered among those.

 For an arbitrary node $v$ of $\protocolPiAlt$ let $D_v \coloneqq \dng(\fix(X_v)) \cup \dng(\fix(Y_v))$. Then the event where one of the block $\bm a, \bm b$ is never fixed in the run of the protocol is expressed simply as ``$\{\bm a, \bm b\} \not\subseteq D_{\bm \ell}$''. Then
 \begin{align} \Pr[(\bm a \bm x_{\bm a} \bm y_{\bm a}, \bm b \bm x_{\bm b} \bm y_{\bm b}) \not\in E(G)] &\le \Pr[(\bm a \bm x_{\bm a} \bm y_{\bm a}, \bm b \bm x_{\bm b} \bm y_{\bm b}) \not\in E(G) \mid \{\bm a, \bm b\} \not\subseteq D_{\bm \ell}] \label{l:first-summand}\\
 &+ \Pr[(\bm a \bm x_{\bm a} \bm y_{\bm a}, \bm b \bm x_{\bm b} \bm y_{\bm b}) \not\in E(G) \land \{\bm a, \bm b\} \subseteq D_{\bm \ell}] \label{l:second-summand}
 \end{align}

\myparagraph{Bounding the summand \eqref{l:first-summand}.}
Observe that for each leaf $\ell$ of $\protocolPiAlt$ the probability $\Pr[\{\bm a, \bm b\} \subseteq D_{\ell} \mid \bm \ell = \ell] \in \{0,1\}$, i.e., leaves can be categorized into safe and dangerous: $\ell$ is \emph{safe} if $\Pr[\{\bm a, \bm b\} \subseteq D_{\bm \ell} \mid \bm \ell = \ell] = 1$, otherwise it is \emph{dangerous}. We also note the following useful claim.

\begin{claim}\label{claim:some-unfixed}
Suppose $\ell$ is a safe leaf, then
        \begin{equation}
            \Pr[(\bm a \bm x_{\bm a} \bm y_{\bm a}, \bm b \bm x_{\bm b} \bm y_{\bm b}) \not\in E(G)  \mid (\bm x, \bm y) \in R_\ell] \le s/n^{0.45}. \label{eq:safe-leaf-bound}
        \end{equation}
\end{claim}

\ifthenelse{\boolean{conferenceversion}}
{}{
\begin{proof}[Proof of \autoref{claim:some-unfixed}]
Wlog assume that $\bm a \notin D_{\bm \ell}$. Fix $\bm x_{\bm a} \bm x_{\bm b} \bm y_{\bm b} = \alpha \in \Sigma^3$ arbitrarily. By the total probability law it suffices to bound the target probability conditioned on each value $\alpha$. By definition of $\dng(\fix(Y_{\bm \ell}))$ we have  $\Hmin(\bm y_{\bm a} \mid (\bm x, \bm y) \in R_\ell) \ge 0.9 \log |\Sigma|$. By $s$-neighbor density we have that $\Pr[(\bm a \bm x_{\bm a}\bm y_{\bm a}, \bm b \bm x_{\bm b} \bm y_{\bm b}) \not\in E(G) \mid (\bm x, \bm y) \in R_\ell, \bm x_{\bm a}\bm x_{\bm b} \bm y_{\bm b} = \alpha] \le |\Sigma|^{-0.9} s \le n^{-0.45} s$. 
\end{proof}
}

Since \eqref{l:first-summand} is a convex combination of left-hand sides of \eqref{eq:safe-leaf-bound}, we conclude by \cref{claim:some-unfixed} that the summand \eqref{l:first-summand} is bounded with $s/n^{0.45}$ as well.

\myparagraph{Bounding the summand \eqref{l:second-summand}.} We would like to apply a union bound over all values of $\{\bm a, \bm b\}$. We first \ifthenelse{\boolean{conferenceversion}}
{verify}{show} this for a \emph{fixed} pair of blocks $a$ and $b$:
\begin{claim} 
\label{claim:all-fixed}
    For every $a \neq b \in [k]$ we have
    \[\Pr[(a \bm x_a \bm y_a, b \bm x_b \bm y_b) \not\in E(G) \mid \{a,b\} \subseteq D_{\bm \ell}] \le s/n^{0.45}.\]
\end{claim}

\ifthenelse{\boolean{conferenceversion}}
{}{
\begin{proof}
Let the run of the protocol $\protocolPiAlt(\xbf,\ybf)$, be the sequence $\vbf_1, \dots, \vbf_d = \bm \ell$ of nodes in the protocol tree.

Consider a node $v$ in $\protocolPiAlt$ corresponding to a rectangle $R_v = X_v \times Y_v \ni (\bm x, \bm y)$ such that $\{a,b\} \subseteq D_v$, but the inclusion is false for the parent of $v$ in $\protocolPiAlt$. Let $T_{i}$ be the set of these nodes in $\protocolPiAlt$ at the depth $i$, and $T \coloneqq \bigcup_{i \le d} T_i$. We then have
 \begin{multline*}
\Pr[(a \bm x_a \bm y_a, b \bm x_b \bm y_b) \not\in E(G) \mid \{a,b\} \subseteq D_{\bm \ell}] = \\\sum_{i \in [k], v \in T_{i}} \Pr[\bm v_i = v] \Pr[(a \bm x_a \bm y_a, b \bm x_b \bm y_b) \not\in E(G) \mid (\bm x, \bm y) \in R_v].
 \end{multline*}
Then it suffices to bound the probability $\Pr[(a \bm x_a \bm y_a, b \bm x_b \bm y_b) \not\in E(G) \mid (\bm x, \bm y) \in R_v]$ for an arbitrary $v \in T$, since the target probability is a convex combination of these probabilities.
Now suppose wlog that the node $v$ was a node where Alice spoke, i.e., for the parent $u$ of $v$ we have $Y_u = Y_v$. Then $\{a,b\} \not\subseteq \dng(\fix(Y_v))$. Indeed, since $v \in T$ we have $\{a,b\} \not\subseteq D_u$, so $\{a,b\} \not\subseteq \dng(\fix(Y_u)) = \dng(\fix(Y_v))$. Assume again wlog that $b \not\in \dng(\fix(Y_v))$. 

Finally, we have by the subcube likeness of $R_v$ and the definition of $\dng(\fix(Y_v))$ that $\Hmin(\bm y_b \mid \bm y \in Y_v) \ge 0.9 \log |\Sigma|$. Then by the graph density $\Pr[(a \bm x_a \bm y_a, b \bm x_b \bm y_b) \not\in E(G) \mid (\bm x, \bm y) \in R_v] \le s/n^{0.45}$ as required. 
\end{proof}}

We need some additional notation to implement our union bound. For a protocol leaf $\ell$ define 
\[p_{a,b}(\ell) \coloneqq \Pr_{\bm z, \bm w \sim R_\ell}[(a \bm z_{a} \bm w_{ a}, b \bm z_b \bm w_b) \not\in E(G)].\]
 Since the pair $(\bm a, \bm b)$ is uniquely determined by $\bm \ell$, so denote $a(\ell)$ and $b(\ell)$ the blocks returned in the leaf $\ell$. 
Having that, we rewrite
\begin{align*}
    \Pr[(\bm a \bm x_{\bm a} \bm y_{\bm a}, \bm b \bm x_{\bm b} \bm y_{\bm b})& \not\in E(G) \land \lnot F] \\
   \text{\small (def. of dangerous leaves) } &=
    \sum_{\ell\colon \text{ dangerous leaf}} p_{a(\ell),b(\ell)}(\ell) \Pr[(\bm x, \bm y) \in R_\ell]\\
    &=
    \Exp[p_{\bm a, \bm b}(\bm \ell) \cdot \mathbbm{1}_{\{\bm a, \bm b\} \subseteq D_{\bm \ell}}]\\
   \text{\small (sum includes the term $\bm a, \bm b$) } &\le \sum_{a \neq b \in [k]} \Exp[\mathbbm{1}_{\{a,b\} \subseteq D_{\bm \ell}} \cdot p_{a,b}(\bm \ell)]\\
    &= \sum_{a \neq b \in [k]} \Pr[\{a,b\} \subseteq D_{\bm \ell} \,\land\,(a \bm x_a \bm y_a, b \bm x_b \bm y_b) \not\in E(G)]\\
    &= \sum_{a \neq b \in [k]} \Pr[(a \bm x_a \bm y_a, b \bm x_b \bm y_b) \not\in E(G) \mid \{a,b\} \subseteq D_{\bm \ell}] \cdot \Pr[\{a,b\} \subseteq D_{\bm \ell}]\\
   \text{\small (number of pairs is $\le |D_{\bm\ell}|^2$) } &\le \Exp[|D_{\bm \ell}|^2] \cdot \max_{a \neq b \in [k]}  \Pr[(a \bm x_a \bm y_a, b \bm x_b \bm y_b) \not\in E(G) \mid \{a,b\} \subseteq D_{\bm \ell}] \\ 
   \text{\small (\autoref{claim:all-fixed}) }&\le d^2 s / n^{0.45}.
\end{align*}

Finally, we have established that the success probability of $\protocolPiAlt$ is at most $2 d^2 s/ n^{0.45} \ge 1 - 2\epsilon = \Omega(1)$, so $d = \Omega(\sqrt{n^{0.45} / s})$ as required.
\section*{Acknowledgements}
\label{sec:acknowledgements}
\addcontentsline{toc}{section}{\nameref{sec:acknowledgements}}
We would like to thank 
Dmitry~Itsykson and 
Kilian~Risse for insightful  
discussions.

David Engström and Susanna F.\ de Rezende received funding from the Knut and Alice Wallenberg grant \mbox{KAW 2023.0116}, ELLIIT, and the Swedish Research Council grant 
\mbox{2021-05104}.
Yassine~Ghannane and Duri~Andrea~Janett were supported by
the Independent Research Fund Denmark grant \mbox{9040-00389B}.
Artur~Riazanov is supported by Swiss
State Secretariat for Education, Research, and Innovation (SERI) under contract number MB22.00026.
We also gratefully acknowledge that we have benefited greatly from being part of Basic Algorithms Research Copenhagen (BARC) environment financed by the Villum Investigator grant~54451.

\bibliography{refs}
\bibliographystyle{plainurl} 

\end{document}